\documentclass[twocolumn]{autart}    

\usepackage{graphicx}          
\usepackage{tabularx}
\usepackage{booktabs}
\usepackage{stfloats}
\usepackage[utf8]{inputenc}
\usepackage[sort&compress,numbers]{natbib}
\usepackage{amsmath}
\usepackage{ntheorem}
\usepackage{amssymb}
\usepackage{appendix}
\usepackage[dvipsnames]{xcolor}
\usepackage{subfigure}
\newtheorem{theorem}{Theorem}
\newtheorem{lemma}{Lemma}
\newtheorem{definition}{Definition}

\newtheorem{assumption}{Assumption}
\newtheorem{remark}{Remark}

\theoremstyle{nonumberplain}
\newtheorem{proof}{Proof}
\graphicspath{{Figs/}}
\allowdisplaybreaks
\begin{document}

\begin{frontmatter}

\title{Distributed energy control in electric energy systems
\thanksref{footnoteinfo}} 

\thanks[footnoteinfo]{This paper was not presented at any IFAC 
meeting. Corresponding author Rupamathi Jaddivada. Tel. +1-412-378-1479}

\author[Jaddivada]{Rupamathi Jaddivada}\ead{rjaddiva@mit.edu},    
\author[Ilic]{Marija D. Ilic}\ead{ilic@mit.edu},               

\address[Jaddivada]{Research Laboratory of Electronics, 10-034, 222 Memorial Drive, Cambridge, MA-02139, USA}  \address[Ilic]{Laboratory for Information and Decision Systems, 32-D726, 32 Vassar Street, Cambridge, MA-02139, USA}  


\begin{abstract}                          
The power interactions of any component in electric energy systems 
with the rest of the system happen naturally, as governed by the energy conservation principles. 
There may, however, occur instances when the rate at which power gets generated by one component through local energy conversion is not exactly the same as that absorbed by rest of the system. This is when instabilities get induced. 
To model and control such instabilities, this paper generalizes  the notion of \textit{interaction variable} used to characterize  diverse system components in a unified manner. 
The same variable captures aggregate system-wide effects and sets reference points for multi-layered distributed output feedback control.
It has a physical interpretation of instantaneous power and generalized reactive power. 
The higher layer design utilizes the interactive energy state-space model to derive intermediate reactive power control, which becomes a control command to the lower layer physical model.
This command is implemented using either Feedback Linearizing Control (FBLC) or Sliding Mode Control (SMC), for which sufficient stability conditions are stated. 
This paper claims that the proposed design is fundamental to aligning dynamic interactions between components for stability and feasibility. 
Without loss of generality, we utilize a simple RLC circuit with a controllable voltage source for illustrations, which is a simplified representation of any controllable component in microgrids.

\end{abstract}

\end{frontmatter}

\section{Introduction}
\label{Sec:Intro}
\begin{table*}
\caption{\textbf{Commonly used notation}}
\begin{tabular}{p{0.015\linewidth}p{0.25\linewidth}p{0.015\linewidth}p{0.25\linewidth}p{0.015\linewidth}p{0.25\linewidth}}
\toprule
\multicolumn{6}{l}{\textbf{Variables}}\\
$x$ & State variable & 
$u$ & Control &
$r$ & port input \\
$m$ & Exogenous disturbance &
$y$ & Output of interest & 
$\tilde{x}$ & Extended state $\left[x,r,m\right]^T$\\
$w$ & \multicolumn{5}{l}{Dual variable associated with control($w^u$), disturbance ($w^m$) or port input ($w^r$)}\\
\multicolumn{6}{l}{\textbf{Superscripts}}\\
$r$ & Interaction port &
$u$ & Control port &
$m$ & Disturbance port \\
$\text{out}$ & Outgoing &
$\text{in}$ & Incoming &
$\text{ref}$ & Setpoint for tracking \\
\multicolumn{6}{l}{\textbf{Subscripts}}\\
$i$ & General component index &
$j$ & Neighboring component index &
$z$ & Variables in energy space \\
\multicolumn{6}{l}{\textbf{Energy space variables (Definitions available in Appendix \ref{Sec:EnergyDef})}}\\
 $P$  & \multicolumn{5}{l}{Instantaneous power
 (Definition \ref{Defn:RealPower}) Eg: $P^u = u w^u$} \\
 $\dot{Q}$  & \multicolumn{5}{l}{Rate of change of reactive power
 (Definition \ref{Defn:ReacPower}) Eg: $\dot{Q}^u = u \dot{w}^u - \dot{u} w^u$} \\
 $E$  & \multicolumn{2}{l}{Stored energy (Definition \ref{Defn:StoredEnergy})}&
 \multicolumn{1}{r}{$p$}  & \multicolumn{2}{l}{Time derivative of stored energy ($\dot{E}$)}\\
 \multicolumn{1}{r}{$E_t$}  & \multicolumn{2}{l}{Stored energy in tangent space (Definition \ref{Defn:StoredEnergyTangent})}&
 $\tau$  & \multicolumn{2}{l}{Time constant (Definition \ref{Defn:TimeConstant})}\\
 $x_z$  & \multicolumn{2}{l}{State variable in energy space $\left[E, p\right]^T$}&
 \multicolumn{1}{r}{$z$}  & \multicolumn{2}{l}{Interaction variable $\left[P,Q\right]^T$ (Definition \ref{Defn:IntVar})}\\
 \bottomrule
\end{tabular}
\label{Table:Notation}
\end{table*}


This paper considers the problem of distributed modeling and control design for the changing electric energy systems. 
These are large-scale complex systems for which control is needed to manage solar power, wind power and load disturbances. 
To pursue distributed control, we characterize components by their internal dynamics as well as by their interactions with the neighboring components. 
This is done by establishing physics-based structured modeling to manage complexity. 
We formalize the inherent multi-layered structures mathematically and build on these structures for deriving the multi-layered control design. 

From the component's standpoint, the control objective is to deliver the required power. However, as system conditions vary, this requires adjusting voltages in response to current deviations and vice-versa. Such adjustments lead to persistent time-varying disturbances, which appear as an incremental negative resistance characteristic undesirable for stable operation. Such disturbances have been a cause of perpetual concern for system stability \cite{bottrell2013dynamic,emadi2006constant,dragivcevic2015dc}. 
These problems are growing due to high penetration of intermittent power plants with low inertia. 

One of the several approaches to control energy systems subject to time-varying disturbances is to design robust centralized controllers immune to any disturbance in a pre-characterized set of disturbances. 
However, in this paper, we seek distributed controllers due to privacy restrictions and their implementation without requiring fast communications. 
Most existing distributed control design methods rely on time-scale separation between the models and control design objectives utilized at multiple control layers \cite{dorfler2015breaking,mohamed2011hierarchical}.  
When the assumptions on time-scale separation fail, each component gets affected by the net power disturbance resulting from interaction with the rest of the system. 
Such disturbances can not be treated as independent inputs since they are fast state-dependent variations.

We propose in this paper a novel modeling approach that lends itself to distributed control design through the exploitation of inherent structures resulting from energy conservation principles. 
The proposed method relies on information exchange between neighboring components. This provides
the information about the interactive dynamics as power is exchanged back and forth between the components and its environment.
Notably, this is a new area of research since the derived models  are not in standard state-space form. 
We provide sufficient mathematical  conditions for feasibility and stability  in terms of  physically intuitive energy and power flow variables. 

The typical design proposed in the literature is based on either hard-to-implement centralized controllers or on entirely local disturbance rejection-based controllers, leading to increased control effort and lack of robustness \cite{chang2014active,izquierdo2010electrical}.
While facilitating coordinated control, the existing distributed controllers require restrictive conditions on internal stability, such as those of zero-state detectability \cite{hill1976stability,wu2019robust}. 
It is therefore, hard to obtain general feasibility, stability, and robustness conditions due to the non-linearity of the resulting interconnected systems governed by general energy conservation principles. 

As a solution to the limitations mentioned above, \cite{ilic1993simple,ilic2007hierarchical} proposed a modeling and control that exploits the inherent structural properties of physical systems. In particular, this work showed that there exists a transformation of internal variables into an \textit{interaction variable}, which can be leveraged for the local control design to have provable distributed stabilization of the interconnected system. 
Interaction variable is defined irrespective of the complex internal energy conversion processes by utilizing the inherent structure of interactions across components in the system, resulting from the first energy conservation law.  
Over the last two decades, the interaction variable-based theory has been extended and utilized to design other types of nonlinear controllers, including control of power-electronics-based inverters, induction machines, generators, and residential demand-side technologies \cite{miao2020high,ilic2020plug,ilic2019toward}. 
In this paper, we extend the previously proposed interaction variable-based approach to accommodate the interactive time-varying disturbances.
\textit{Contributions and outline:}
In Section \ref{Sec:Formulation}, we formally pose the distributed control problem as a tracking and output-disturbance decoupling problem from the component's perspective. 
We present a general multi-layered modeling framework that helps re-pose the distributed control problem drawing on the inherent intuitive structure resulting from the underlying physics, which is summarized in Section \ref{Sec:Modeling}. 
To do so, we extend the previously proposed interaction variable-based approach introduced  in \cite{ilic1993simple, ilic1996hierarchical}
by defining an intuitive two-dimensional common output variable of interest and its dynamics based on fundamental energy conservation principles. 
We thereby obviate the need to perform complex coordinate transformations searching for common output variables that would otherwise be needed \cite{decarlo1988variable, guo2000nonlinear}.
Utilizing this modeling, we present an overview of the interactive control in Section \ref{Sec:InteractiveCtrl}. In particular, we show 
how the selection of common variables in energy space
assists in designing a control cognizant of both component-level stability conditions resulting from the internal dynamics and system-level feasibility conditions pertaining to the interconnection dynamics, derived in Section \ref{Sec:Control} and \ref{Sec:Feasibility} respectively. 
While these are only sufficient conditions, they have far reaching implications since they
offer intuitive understanding in terms of energy and power flows and are modular. 
In Section \ref{Sec:Control}, 
we extend the control design to an easy-to-implement switching control robust to modeling and parameter uncertainties. This approach notably supports data-enabled decision-making due to its minimal dependence on internal physics-based models. 
We take an example of an RLC circuit in Section \ref{Sec:RLCCkt} and illustrate the performance and robustness of our proposed energy control.
Finally, we provide some concluding thoughts in Section \ref{Sec:Conclusion}.

\section{Distributed control problem formulation}
\label{Sec:Formulation}
Without loss of generality, consider a simple electrical energy system comprising two components as shown in Fig. \ref{fig:system_general}. 
In order to formally pose the distributed control problem
we first devise the component models.  One approach to obtaining these models is tearing of the interconnected system model by defining additional variables at the intersection of two components, which are also referred to as port inputs, denoted as $r_i$. These inputs dictate the interaction of the component $i$ with the rest of the system \cite{willems2007behavioral}. \begin{figure}[!htbp]
\begin{center}
\includegraphics[width=1.0\linewidth]{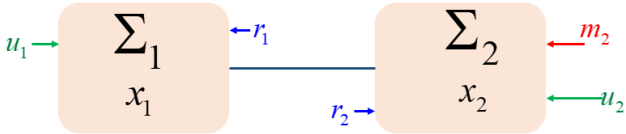}
\caption{Interconnected system comprising two components $\Sigma_1$ and $\Sigma_2$ with local controllable input $u_1, u_2$ and exogenous disturbances $m_2$}
\label{fig:system_general}
\end{center}
\end{figure}
Algebraic relations such as Kirchoff's current and voltage laws dictate the relationship between the port inputs and state variables. 
This relation is abstracted through a map $L$ as shown in Eqn. \eqref{Eqn:Algebraic_Conv} \cite{decarlo1984component,bachovchin2015design}.
\begin{equation}
    0 = L(x_1,x_2, r_1, r_2)
    \label{Eqn:Algebraic_Conv}
\end{equation}
A schematic of component models comprising the interconnected system is shown in Figure \ref{fig:system_general} and the corresponding dynamical models are given in Eqn. \eqref{Eqn:GeneralCompModel_COnv}. 
\begin{subequations}
\begin{align}
&\text{State Dynamics:} \qquad\qquad\qquad\qquad x_i(0) = x_{i,0}\notag\\
&\dot{x}_i = f_{x,i}(x_i) + g_i^r(x_i) r_i + g_i^m(x_i)m_i + g_i^u(x_i) u_i  \label{Eqn:StateGeneralDyn}\\
&\text{Outputs of interest}:
y_i = f_{y,i}(x_i, u_i, r_i, m_i)  \\
&\text{Common outputs}:
    z_i = f_{z,i}(x_i, u_i, r_i, m_i)    
\end{align}
\label{Eqn:GeneralCompModel_COnv}
\end{subequations}
The notation used here is summarized in Table \ref{Table:Notation}.


The major difficulty in distributed control problem is to find the common output variable $z_i$. 
Assuming such common output variable can be found, the control problem can be posed as follows: \\
\textit{\textbf{(P1):}
The objective is to design a smooth feedback control $u_i = k_i(y_i, z_j) ~\forall i \in \{1,2\}, j \ne i$ utilizing local information ($y_i$) and minimal information sent by neighboring component $\Sigma_j$ ($z_j$),  to satisfy following objectives:
\begin{enumerate}
    \item Dissipativity: Each component $\Sigma_i$ in closed loop is dissipative \cite{willems2007behavioral} w.r.t. a supply function $w_i\left(\left[m_i,r_i\right], y_i\right)$.
    \item Internal stability: It is further desired of the supply function to satisfy $w_i\left(\left[m_i,r_i\right], y_i\right) \le 0 \quad \forall y_i$ and admissible component inputs $r_i$ and $m_i$.
    \item Output regulation: Local outputs of interest $y_i(t)$ must 
    be regulated to a fixed reference value $y_i^{ref}$
    \item Feasibility: The common output variables track a consistent reference given by smooth map $\psi_i$ as $z_i^{ref}(t) = \psi_i(x_i(t), z_j(t)) ~~ \forall j \in \mathcal{C}_i$ where $\mathcal{C}_i$ represents the set of component indexes corresponding to the neighbors. 
\end{enumerate} 
}

The objectives stated in (P1) are multi-fold including stabilization of internal dynamics, tracking of common output variables and the robustness to local exogenous disturbances. 
It is apparent that the control objectives above can be contradicting each other. This problem only gets more complex with the different types of disturbances entering the system. 


\section{Interactive energy-based modeling}
\label{Sec:Modeling}
Next, we introduce a common output variable that helps solve the different control objectives posed in (P1).
We refer to  this common output variable as an interaction variable. 
Its physical meaning is  the energy and  reactive power absorbed by the component by virtue of its own energy conversion processes.
More fundamentally, it is defined as follows: 

\begin{definition}(Interaction Variable) \cite{ilic1993simple,jaddivada2020unified} 
\footnote{This definition was provided in particular for electric power systems under an assumption of real-reactive power decoupling \cite{ilic1993simple,ilic2007hierarchical}. We now further extend this notion by relaxing the decoupling assumption.}
\\
The interaction variable $z_i^{r,out}$ is defined as a function of local variables that satisfies the property 
\begin{subequations}
\begin{equation}
z_i^{r,out} = \text{constant}
\end{equation}
when all interconnections are removed. \\
Let $E_i$, $p_i$, $P_i^u$, $P_i^m$, $\dot{Q}_i^u$, $\dot{Q}_i^m$ represent the stored energy, rate of change of stored energy, instantaneous power at control terminal and disturbance terminal and generalized rate of reactive power at control and disturbance terminals respectively of component $i$.
Each of these variables is defined as a function of local state variables and state derivatives in Appendix \ref{Sec:EnergyDef}. 
Mathematically, the interaction variable is then defined as 
\begin{equation}
	\begin{array}{*{20}{l}}
{z_i^{r,{\rm{out}}} = \left[ {\begin{array}{*{20}{l}}
{\int\limits_0^t {\left( {{p_i}(s) + \frac{{{E_i}(s)}}{{{\tau _i}}} - P_i^u(s) - P_i^m(s)} \right)ds} }\\
{\int\limits_0^t {\left( { - {{\dot p}_i}(s) + 4{E_{t,i}}(s) - \dot Q_i^u(s) - \dot Q_i^m(s)} \right)ds} }
\end{array}} \right]}
\end{array}
		\label{Eqn_IntModelInv}
		\end{equation}
\end{subequations}
\label{Defn:IntVar}
\end{definition}

The existence and the defining structural properties of the interaction variable were proved in \cite{jaddivada2020unified}.
In order to differentiate the interactions resulting from internal energy conversion processes as per the Definition \ref{Defn:IntVar}, and the ones obtained as a result of interconnection, we utilize the superscripts `$\rm{out}$' and `$\rm{in}$' respectively.
The incoming interaction variable is a result of interconnection, as shown in the zoomed-out representation of the interconnected system in Fig. \ref{Fig:Cutset}. 
\begin{figure}[!htbp]
\begin{center}
\includegraphics[width=1.0\linewidth]{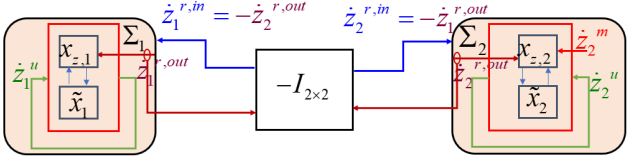}    
\caption{Zoomed-out representation for the interconnected system: 
Incoming interactions ($z_i^{r,in}$) are shown with blue arrows, while the outgoing ones ($z_i^{r,out}$) by virtue of local energy conversion dynamics are shown in brown for each of the components in the closed-loop.
After interconnection, the incoming interaction variable is equal to the negative of outgoing interaction variables of its neighbors.
}  
\label{Fig:Cutset}                                 
\end{center}                                 
\end{figure}

We next introduce the zoomed-in component models. The physical models introduced in Eqn. \eqref{Eqn:GeneralCompModel_COnv} are related to the interaction variable and its dynamics, through a bi-directional dynamic mapping as shown in Fig. \ref{CptModel_EnergySpaceInteractive}. 
\begin{figure}[!htbp]
	 	\begin{center} 
	 	\includegraphics[width=1.0\linewidth]{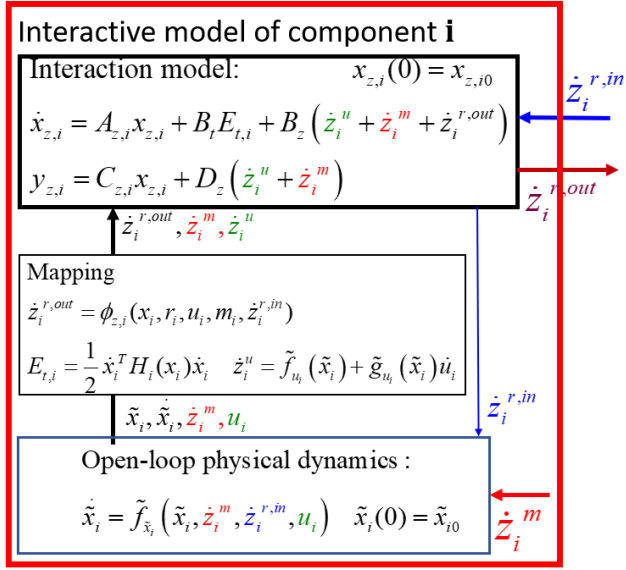}
	 	\caption{Interactive stand-alone model of an open-loop component $i$: 
	 		The lower layer models are utilized to compute the outgoing interaction variable $\dot{z}_i^{r,\rm{out}}$, which drive the higher-layer energy dynamics of the component.  The incoming interaction variable from the grid $\dot{z}_i^{r,\rm{in}}$, is utilized by the lower-layer models to evaluate the extended state trajectories $\tilde{x}_i = \left[x_i,m_, r_i\right]$ given their initial conditions. 
	 	}		
	 	\label{CptModel_EnergySpaceInteractive}
	 	\end{center}
\end{figure}

From this figure, it is important to note that the physical lower layer model is characterized now using extended state space $\tilde{x}_i = \left[x_i, m_i, r_i\right]^T$.
Here, we characterize the persistent disturbances and incoming interactions at the ports through the instantaneous power and generalized reactive power at ports by the vector $\dot{z}_i^m$ and $\dot{z}_i^{r,in}$ respectively.
In addition to the state dynamics, it now becomes imperative to model the dynamics of the disturbance and port inputs in order to include the dependence of instantaneous and reactive power entering the respective terminals ($\dot{z}_i^m, \dot{z}_i^{r,in}$). Characterization of these power variables, in turn, helps retain the linear structure of the higher layer energy space model in Fig. \ref{CptModel_EnergySpaceInteractive}, also elucidated in Eqn. \eqref{Eqn_IntModelLinear}. 
The dynamics of extended state variables are expressed in Eqn. \eqref{Eqn:InteractiveStdAloneModel}. 
\begin{equation}
\begin{array}{l}
{\rm{Extended~ state~ space~ model}}:~{{\tilde x}_i}(0) = {\left[ {\begin{array}{*{20}{c}}
{{x_{i0}}}&{{m_{i0}}}&{{r_{i0}}}
\end{array}} \right]^T}\\
\underbrace {\left[ {\begin{array}{*{20}{l}}
\begin{array}{l}
{{\dot x}_i}\\
{{\dot m}_i}
\end{array}\\
{{{\dot r}_i}}
\end{array}} \right]}_{{{\dot{\tilde{x}}}_i}} = \underbrace{\begin{array}{l}
\left[ {\begin{array}{*{20}{l}}
\begin{array}{l}
{f_{x,i}}({x_i}) + {g_i}^m({x_i}){m_i} + {g_i}^r({x_i}){r_i}\\
{f_{m,i}}({{\tilde x}_i}) + g_{m,i}^m({{\tilde x}_i}){m_i} + {g_{m,i}}^r({{\tilde x}_i}){r_i}
\end{array}\\
{{f_{r,i}}({{\tilde x}_i}) + g_{r,i}^m({{\tilde x}_i}){m_i} + {g_{r,i}}^r({{\tilde x}_i}){r_i}}
\end{array}} \right] + \\
\left[ {\begin{array}{*{20}{c}}
0&0\\
{g_{m,i}^{zm}({{\tilde x}_i})}&0\\
0&{g_{r,i}^{zr}({{\tilde x}_i})}
\end{array}} \right]\left[ \begin{array}{l}
{{\dot z}_i}^m\\
{{\dot z}_i}^{r,in}
\end{array} \right] + \left[ {\begin{array}{*{20}{l}}
\begin{array}{l}
{g_i}^u({x_i})\\
g_{m,i}^u({{\tilde x}_i})
\end{array}\\
{g_{r,i}^u({{\tilde x}_i})}
\end{array}} \right]{u_i}
\end{array}}_{{{\tilde f}_{{{\tilde x}_i}}}\left( {{{\tilde x}_i},{{\dot z}_i}^m,{{\dot z}_i}^{r,in},{u_i}} \right)}
\end{array}
\label{Eqn:InteractiveStdAloneModel}
\end{equation}
Here $f_{x,i}, g_i^r, g_i^m$ are the same functions defined in Eqn. \eqref{Eqn:StateGeneralDyn}. 
The rest of the functions are a result of expressions for reactive power at respective ports according to Definition. \ref{Defn:ReacPower}. 

\subsection{Distributed control problem posing in energy space}
We propose to characterize the output variables $y_i$ and $z_i$ in energy space for obtaining an intuitive map $\psi_i$ in the fourth objective of tracking in problem (P1) for re-interpreting the distributed control problem posed for a feasible and stable interconnected system.
We henceforth use the notation $y_{z,i}$ and $\dot{z}_i^{r,out}$ characterized in energy space to represent the output variables $y_i$ and $z_i$ respectively in the problem formulation (P1). 
The model for use with the problem formulation (P1) in transformed state space is shown in Eqn. \eqref{Eqn_IntModelLinear}. 
Here, the state variables are aggregate dynamical energy variables denoted as $x_{z,i} = \left[E_i, p_i\right]^T$.
\begin{subequations}
		\begin{align}
		&\text{Energy space state dynamics:} \qquad x_{z,i}(0) = x_{z,i0} \label{Eqn:IntModelState}\\
		&\dot{x}_{z,i} = A_{z,i} x_{z,i} + B_t E_{t,i}\left(\dot{x}_i\right) +B_{z} \left(\dot{z}_i^{r,out} + \dot{z}_i^u + \dot{z}_i^m\right)\notag\\
		& \text{Rate of change of common outputs: } ~~ z_{i}^{r,out}(0) = z^{r,out}_{i0} \notag\\
		&\dot{z}_i^{r, out} = \phi_{z,i}(x_i, r_i, u_i, m_i, \dot{z}_i^{r,in}) \label{Eqn:IntModelCommon}\\
		& \text{Outputs of interest: }  \notag\\
		&{y_{z,i}} = \underbrace {\left[ {\begin{array}{*{20}{c}}
{\frac{1}{{{\tau _i}}}}&0
\end{array}} \right]}_{{C_{z,i}}}{x_{z,i}} + \underbrace {\left[ { - 1\;\;0} \right]}_{{D_z}}\left( {{{\dot z}_i}^u + {{\dot z}_i}^m} \right)
\label{Eqn:IntModelOutput_v1}
		\end{align}
		\label{Eqn_IntModelLinear}
		\end{subequations}
The energy state space evolution is given in 
Eqn. \eqref{Eqn:IntModelState}, which was first postulated in \cite{ilic2018multi} as a generalization of energy conservation principles.
In Eqn. \eqref{Eqn_IntModelLinear}, the constant matrices and vectors used in the model are: 
		$B_{t} = \left[0, 4\right]^T,  B_z = \left[1~~ -1\right]^T  
		$ 	for any component and matrix $A_{z,i} = \left[ {\begin{array}{*{20}{c}}
			 - 1/{\tau_i}&0\\
			0&0
			\end{array}} 
		\right]$ 
			depends only on the time constant $\tau_i$ defined in Definition \ref{Defn:TimeConstant}.

We further define the common output variable as the interaction variable defined in Eqn. \eqref{Eqn_IntModelInv}. This variable can only be numerically computed by utilizing the analytical expressions given for the  rate of change of interaction variable 
through an abstract map $\phi_{z,i}$ of internal variables as shown in Eqn. \eqref{Eqn:IntModelCommon} \cite{jaddivada2020unified}. 
Since outgoing interaction variable by Definition \ref{Defn:IntVar} is a function of local states and state derivatives (function of port input $r_i$), the rate of change of outgoing interaction variable depends on the rate of change of incoming interaction variable. 
Such dependence makes the modeling framework inherently interactive. 
For more details on the interactive modeling approach, the reader is referred to \cite{ilic2018multi, ilic2018fundamental, jaddivada2020unified}.

We propose to characterize the output variable of interest to have an interpretation of power produced or absorbed by the component. 
Over relatively slower time scales, instantaneous power balance or reactive power balance can be set up as a control objective since one implies the other over these time scales. 
Furthermore, if only one physical control input is available, only one out of the two balance equations can be set as a control objective. 
We thus select the output variable of interest as the instantaneous power absorbed the component after energy dynamics settle i.e. after $\dot{E}_i = 0$. 

The control problem ${(P1)}$ can better be posed in reduced order linear energy space. 
However, the model is dependent on the internal dynamics and vice-versa as shown in Fig. \ref{CptModel_EnergySpaceInteractive}.
This notion of interactive modeling is also inherent in Willems' seminal paper on behavioral modeling of physical systems \cite{willems2007behavioral}. While standard state space modeling has been central to the development of control theory for decades, we emphasize that interactive modeling is instrumental for distributed implementation. Next, utilizing these interactive models, we propose a multi-layered interactive control design to solve the problem (P1).

\section{Multi-layered distributed control design}
\label{Sec:InteractiveCtrl}
In this section, we provide an overview of the proposed multi-layered control architecture, that utilize the inherent structure of linear energy space models for a provable control design. The proposed method constitutes two parts of the design summarized as follows:
\begin{itemize}
    \item Component-level: Given the incoming interaction variable $\dot{z}_i^{r,in}$ from the interconnection-level control, 
    map in Eqn. \eqref{Eqn:PsiMap} is utilized to obtain the reference point $y_{z,i}^{ref}$. The control objective of this layer is to ensure that the outputs of interest $y_{z,i}$ track the feasible reference point $y_{z,i}^{ref}$.  This is achieved by two-layered control at the component level as follows: 
    \begin{itemize}
        \item The linearity of the higher layer energy-space model shown in Fig. \ref{CptModel_EnergySpaceInteractive} is utilized to design a control in energy space that ensure $y_{z,i}$ tracks $y_{z,i}^{ref}$, 
        \item The lower layer implementation is then performed through a diffeomorphic map between the energy space variables and the conventional space variables at the control ports.
    \end{itemize}
    \item Interconnection-level: Given the ranges of outgoing interaction variable bottom-up by the components for a period of time, 
    the zoomed-out interconnected system model depicted through Fig. \ref{Fig:Cutset} is utilized to obtain optimal values of incoming interaction variable that need to be distributed among multiple components, so as to guarantee feasibility and efficient utilization of available controllers. 
    For a two-component interconnection of interest in this paper, optimal computation of incoming interaction variable over a period of time is not necessary. 
    We instead propose a simple modular feasibility conditions that each component in closed-loop can check for before getting connected to the interconnection with the grid. For a complete optimal control problem formulation, the interested readers are referred to \cite{ilic2020unified}. 
\end{itemize}
The feasibility of a component's interconnection with the rest of the system is dictated primarily by the existence of a solution for the interface variables dictated by Kirchoff's current or voltage laws. 
The port current and voltage variables can enter the component dynamical models non-linearly that makes the problem of answering the feasibility questions quite involved \cite{ilic2018fundamental}. 
It has been shown in \cite{ilic2018multi} that instantaneous power and generalized reactive power balance equations shown through the equality constraints in Fig. \ref{Fig:Cutset} imply satisfaction of Kirchoff's current and voltage laws both statically and dynamically. 
As a result, we can better answer feasibility questions under this new modeling framework, especially when systems are subject to persistent disturbances. 

The reference values of the outputs should then be equal to incoming interaction variable so that instantaneous power balance equation over relatively slower timescales is satisfied.  
In terms of the feasibility objective in problem (P1), 
an intuitive mapping $\psi_i: (x_i, \dot{z}_j^{r,out}) \rightarrow y_{z,i}^{ref} ~~\forall j \in \mathcal{C}_i$ is proposed as:
\begin{equation}
   y_{z,i}^{ref} = \left[ {1\;\;0} \right]\dot z_i^{r,in} =  - \left[ {1\;\;0} \right]\sum\limits_{j \in {{\mathcal{C}}_i}} {\dot z_j^{r,out}}  = {D_z}\sum\limits_{j \in {{\mathcal{C}}_i}} {\dot z_j^{r,out}} 
    \label{Eqn:PsiMap}
\end{equation}
The second equality is a result of the application of generalized Tellegen's theorem \cite{ilic2018multi,penfield1970generalized}. 
Incorporating such specifications in the control design enables the component to generate or absorb energy at a rate equal to the one that the rest of the system injects into the component. 
Such an objective, when incorporated by all components in the system, 
leads to the components interactively achieving consensus on power-sharing, thereby ensuring a feasible interconnection. 
However, further investigation is needed to ensure the existence and reachability of a feasible equilibrium through proper control design. 
In particular, we have only addressed the fourth objective stated in problem (P1) by defining a consistent map $\psi_i$. 
We focus on the second and third objectives in Section \ref{Sec:Control}, and the first objective in Section \ref{Sec:Feasibility}. 

Before formalizing the claims in the rest of the paper, we make the following assumptions that typically hold in electric energy systems. 
\begin{assumption}
	The interactive model of the sub-system $\Sigma_i$ as defined in Fig. \ref{CptModel_EnergySpaceInteractive} has sufficient smoothness conditions to make the sub-system well defined. 
	More precisely, for any $\tilde{x}_{i,0}$ and admissible pair of $\dot{z}_i^{m}(t), \dot{z}_i^{r,\rm{in}}(t)$, there exists at least one  solution $\tilde{x}_i(t) \forall t \in [0, \infty)$ so that the output of interest $y_{z,i}$ is locally square integrable.
	\label{Assum_Integration}
\end{assumption}
\begin{assumption}
	The extended state space of $\Sigma_i$ is reachable from origin. More precisely, given any $\tilde{x}_i$ and $t_1$, there exists a $t_0 \le t_1$ and an admissible incoming interaction $\dot{z}_i^{r,in}(t)$ and local disturbances $\dot{z}_i^ m$ such that the extended states can be driven from $[x_i(t_0), r_i(t_0)] = {0}$ to $[x_i(t_1), r_i(t_1)] = [x_i, r_i]$
	\label{Assum_Reach}
\end{assumption}
\begin{assumption}
	The stored energy in tangent space $E_{t,i}$ as stated  in Definition \ref{Defn:StoredEnergyTangent} is 
	positive at all times. 
	\label{Assum_Energy}
\end{assumption}
\begin{assumption}
    The stored energy function $E_i(x_i)$ is a positive definite function and 
	the time constant $\tau_i$ as defined in Definition \ref{Defn:TimeConstant} is assumed positive at all times.
	\label{Assum_TimeConstant}
\end{assumption}
In Assumption \ref{Assum_Integration}, the admissible inputs refer to the class of inputs belonging to a set of finite energy signals, i.e. they must remain square integreable.
We can also characterize the set with other additional constraints such as those corresponding to physical equipment protection limits. 
Assumptions \ref{Assum_Energy} and \ref{Assum_TimeConstant} mean that the components' inherent inertia and damping (See Definitions \ref{Defn:StoredEnergyTangent} and \ref{Defn:TimeConstant}) are positive which is true for most physical systems.

\section{Interactive component-level control}
\label{Sec:Control}
Given a consistent reference value for the output variable, map based on local measurements such as the one obtained by the mapping Eqn. \eqref{Eqn:PsiMap}, we propose a two-layered control that goes hand-in-hand with the interactive models proposed in Section \ref{Sec:Modeling}, to satisfy the control objectives stated in (P1).
By utilizing the energy space model, we first design the control in higher layer to ensure feasibility by ensuring the tracking error $y_{z,i} - y_{z,i}^{ref}$ is driven to zero. 
In this linear energy space layer, we treat $u_{z,i} = \dot{Q}_i^u$ as 
the virtual control input. 

\begin{figure}[ht]
	\centering 
	{\includegraphics[width=1.05\linewidth]{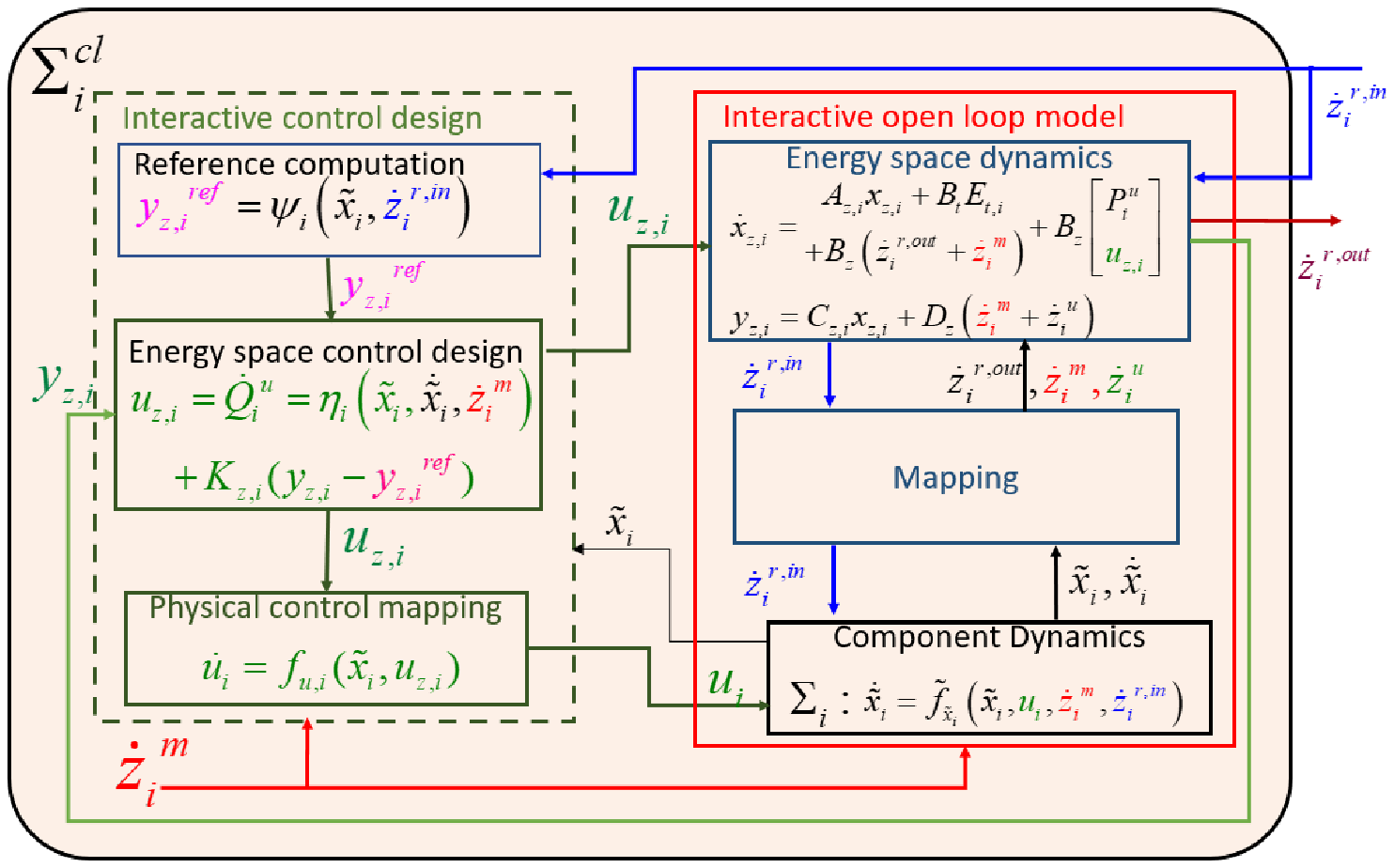}}
	\caption{Block diagram of closed loop module formed by component-level distributed interactive control design satisfying objectives stated in control problem (P1).}
	\label{CptCtrl_EnergySpaceStabCtrl}
\end{figure}
This design is mapped to physical state space 
by utilizing the Definition \ref{Defn:ReacPower}. 
We can design a dynamic controller to be implemented as in Eqn. \eqref{Eqn:ControlImp}. Corresponding dynamic map is abstracted through the function $f_i^u$ also utilized in the control implementation block in Fig. \ref{CptCtrl_EnergySpaceStabCtrl}. 
\begin{equation}
    \dot{u}_i = \frac{\dot{w}_i^u}{w_i^u} u_i - \frac{1}{w_i^u}\dot{Q}^u_{i} = f_i^u(\tilde{x}_i, u_{z,i})
    \label{Eqn:ControlImp}
\end{equation}
Note that the map $f_i^u$ is well defined since $w_i^u$ is non-zero for non-zero higher layer design variables $\dot{z}_i^u = \left[P_i^u, \dot{Q}_i^u\right]$. 
A schematic of the proposed design in shown in Fig. \ref{CptCtrl_EnergySpaceStabCtrl}. 

The proposed two-layered design can also be interpreted as an input-output linearization based control. 
In particular, consider the proposed output variable in Eqn. \eqref{Eqn:IntModelOutput_v1}.
By taking its time derivative and plugging in the relations given by the interaction model in Eqn. \eqref{Eqn:IntModelState}, we have the dynamic input-output relation between the second component of input $\dot{z}_i^u$ and the chosen output $y_{z,i}$ as shown in Eqn. \eqref{Eqn:NormalFOrm}.
\begin{subequations}
\begin{align}
&{{{\dot y}_{z,i}} = - 4{E_{t,i}} +\underbrace {\left( {  {{\dot P}_i}^{r,out} + {{\dot Q}_i}^{r,out} + \dot Q_i^m} \right)}_{{\eta _i}({{\tilde x}_i},{\dot{\tilde{x}}_i},\dot z_i^m)} + \underbrace {\dot Q_i^u}_{{u_{z,i}}}} \label{Eqn:NormalFormOutput}\\
&{{\dot{\tilde x}}_i} = {{\tilde f}_{{{\tilde x}_i}}}({{\tilde x}_i}) + \tilde g_{{{\tilde x}_i}}^m({{\tilde x}_i}){{\dot z}_i}^m + \tilde g_{{{\tilde x}_i}}^r({{\tilde x}_i}){{\dot z}_i}^{r,in} \notag\\
&{\qquad \qquad  + \tilde g_{{{\tilde x}_i}}^u({{\tilde x}_i}){U_i}({{\tilde x}_i},{P_i}^m,{y_{z,i}})} \label{Eqn:NormalFormInt}
\end{align}
\label{Eqn:NormalFOrm}
\end{subequations}
Here, the map $U_i(x_i, P_i^m, y_{z,i}) \rightarrow u_i$ is defined as in Eqn. \eqref{Eqn:ControlMapExp} to obtain physical control $u_i$ given the output variable in energy space. 
\begin{equation}
    {u_i = \frac{1}{{w_i^u}}\left( {\frac{1}{2}{x_i}^T{B_i}({x_i}){x_i} - {y_{z,i}} - {m_i}w_i^m} \right)}
    \label{Eqn:ControlMapExp}
\end{equation}
The input-output model in Eqn. \eqref{Eqn:NormalFOrm} has a relative degree equal to one since the first time derivative of the output is directly dependent on the control in energy space $u_{z,i} = \dot{Q}_i^u$. 
The rest of the state variables characterizing the internal dynamics, are not directly affected by  the virtual control in energy space $u_{z,i}$. 
They are only affected by physical control $u_i$, which is re-expressed through the map $U_i(x_i, P_i^m, y_{z,i})$. 

The normal form model is useful in direct control of outputs of interest, the usage of which
is further warranted by the fact that it is diffeomorphic to the model in Eqn. \eqref{Eqn:StateGeneralDyn} and \eqref{Eqn:ControlImp} \cite{jaddivada2020unified}. 
We next propose two different energy-space control design methods, with qualitatively different implications on the performance. 

\subsection{Feedback linearizing control}
In order to ensure the tracking objective on $y_{z,i}$ following $y_{z,i}^{ref}$,
 we propose control design in energy space as: 
\begin{equation}
    u_{z,i}=  -{\hat{\eta} _i} - K_i(y_{z,i} - y_{z,i}^{ref}) + \dot{y}_{z,i}^{ref}
    \label{Eqn:FBLC}
\end{equation}
Here $K_i$ is a positive constant and 
$\hat{\eta}_i$ is an estimate of  the $\eta_i$ defined in Eqn. \eqref{Eqn:NormalFormOutput}, which undoubtedly is a complex design task.
However, although the definition of terms in use for 
$\eta_i$ involve higher order terms in conventional state space, its estimate can be obtained using the linear energy state space model of Eqn. \eqref{Eqn_IntModelLinear}. This discussion is however out of scope of this paper and preliminary work on this topic can be found in \cite{bharadwaj2021measurement}. For the purposes of this paper, we assume that an accurate estimate of $\eta_i$ can be obtained.
Substituting the control design of Eqn. \eqref{Eqn:FBLC} into Eqn. \eqref{Eqn:NormalFormOutput}, we obtain in closed loop
\begin{equation}
    \frac{d}{dt}\left(y_{z,i} - y_{z,i}^{ref}\right) = -K_i\left(y_{z,i} - y_{z,i}^{ref}\right) + \left(\tilde{\eta}_i- 4E_{t,i}\right)
    \label{Eqn:OutputClosed}
\end{equation}
Here $\tilde \eta_i = \eta_i - \hat{\eta}_i$ is the difference between the physical and measured nonlinearities that are utilized in the higher layer design in energy space.  
From the output response in Eqn. \eqref{Eqn:OutputClosed}, 
we establish 
sufficient stability conditions for the closed loop component models. In order to do so, we need to introduce the notion of dissipativity in the extended state space. 
The definition in \cite{kawano2020krasovskii} has been revised to include the effect of control inputs in extended state space designed using the interactive energy space models introduced in Fig. \ref{Eqn:InteractiveStdAloneModel}. 
\\
\begin{definition}
\textit{(Dissipative through feedback)}
The system in Eqn. \eqref{Eqn:NormalFOrm} is said to be dissipative through feedback with respect to supply rate $s_i$ if there exists a feedback control $u_{z,i}$
and a storage function $S_i$ satisfying
\begin{equation}
\begin{array}{l}
\frac{d{S_i}}{dt}({{\tilde x}_i}(t))  \le {{s_i}\left({y_{z,i}}({{\tilde x}_i}(t)), {\dot{z}_i^{r,in}(t)} \right)} 
\end{array}
\label{Eqn:DiffDissip}
\end{equation}
for all $t \ge 0$ and all $\left(\tilde{x}_i, \dot{z}_i^{r,in}\right) \in \mathcal{\tilde{X}}_i \times \mathcal{Z}_i$  where $\mathcal{\tilde{X}}_i$ and $\mathcal{Z}_i$ represents the extended state space and the incoming interaction variable manifolds respectively. 
\end{definition}$~$
\begin{theorem}(Performance with FBLC control)\\
    \begin{subequations}
    Consider the virtual control $u_{z,i}$ designed as in Eqn. \eqref{Eqn:FBLC}.
    Under assumptions \ref{Assum_Integration} - \ref{Assum_Energy},
    the closed loop model formed by equations
    \eqref{Eqn:NormalFOrm}, \eqref{Eqn:FBLC}
    establish following properties:
    \begin{enumerate}
	    \item 
	    It is dissipative through feedback w.r.t a supply rate
	    \begin{equation}
	    {s_i}\left( {{y_{z,i}},{{\dot z}_i}^{r,in}} \right) =  - K_i \left( {{y_{z,i}} - y_{z,i}^{ref}} \right)  + \dot{P}_i^{r,in} + \dot{Q}_i^{r,in}
	    \label{Eqn:Dissipativity}
	    \end{equation}
	    if $\tilde{\eta}_i \le \dot{Q}_i^{r,in}$
	    where $y_{z,i}$ is a function of extended state space as defined in Eqn. \eqref{Eqn:IntModelOutput_v1} and 
	    $y_{z,i}^{ref} = P_i^{r,in}$ which is one of the elements of $\dot{z}_i^{r,in}$. 
	    \item 
	     Assume there exists a non-empty equilibrium set defined as
	      $   \left\{\tilde{x}_i^* \in \mathcal{\tilde{X}}_i  |y_{z,i}\left(\tilde{x}^*_i\right) = P_i^{r,in}\right\}$
	    Any point in this set is asymptotically stable if
	    \begin{equation}
	    \left|\tilde{\eta}_i - 4E_{t,i}\right| \le  K_i\left|y_{z,i} - y_{z,i}^{ref}\right| 
	    \label{Eqn:SuffStabCond}
	    \end{equation}
    \end{enumerate}
	\end{subequations}
    \label{Theorem:CtrlFBLC}
\end{theorem}
\begin{proof}
The proof is detailed in Appendix \ref{Sec:CtrlFBLCProof}. 
\end{proof}

It should be noted that the equilibrium set is unknown and is itself possibly time-varying. As a result, we do not assume stored energy in tangent space is zero. 
The sufficienet stability condition in Eqn. \eqref{Eqn:SuffStabCond} indicates that the nonlinearity cancellation error needs to be upper bounded by the rate at which output tracking is performed. The stored energy in tangent space further aids in satisfaction of this constraint. 
Upon satisfaction of such condition, the internal dynamics of the input-output linearized model in Eqn. \eqref{Eqn:NormalFormInt} for $y_{z,i} = y_{z,i}^{ref}$ can be deemed to be stable since bounded state trajectories are assumed. 
The overall performance is thereby dominated by the 
feedback linearizing output in Eqn. \eqref{Eqn:OutputClosed}. We can conclude that even for time-varying disturbances, the tracking error decays exponentially.  


\subsection{Sliding mode control}
For provable performance, the feedback linearizing control requires exact cancellation of the nonlinear term $\eta_i$. Furthermore, obtaining the measurements needed for the computation of $\eta_i$ can be difficult in general. 
We thus propose a sliding mode equivalent control
as shown in Eqn. \eqref{Eqn:SMC}, where only bound  $\left|\eta_i - 4E_{t,i} \right|\le \overline{L}_i$ over a pre-specified time needs to be known.   
\begin{equation}
    \dot{Q}_i^u = -\underbrace{\left(\overline{L}_i + K_i\right)}_{\alpha_i} sign\left(y_{z,i}- y_{z,i}^{ref}\right) +\dot{y}_{z,i}^{ref}
        \label{Eqn:SMC}
    \end{equation}
Here, the sliding surface is the tracking error $\left(y_{z,i}- y_{z,i}^{ref}\right)$ and $K_i$ is a positive constant.
By selecting a large enough gain value $\alpha_i$, the sliding mode controller is equivalent to FBLC control \cite{slotine1991applied}. 

\begin{theorem}(Performance with SMC)\\
    \begin{subequations}
    Under assumption \ref{Assum_Integration} - \ref{Assum_Reach} 
    the closed loop model formed by equations \eqref{Eqn:NormalFOrm} and \eqref{Eqn:SMC} 
    establish following properties:
    \begin{enumerate}
	    \item Tracking of output variable $y_{z,i}$ to the reference value $y_{z,i}^{ref} = P_i^{r,in}$ is achieved within a finite reaching time upper bounded by
    $t_r = \frac{{{2}}}{K_i} \left|\sigma_i(0)\right|$. 
    where $\left|\sigma_i(0))\right|$ is the distance of the operating point from the sliding surface and 
    $K_i = \alpha_i - \overline{L}_i$ is a positive constant 
where $\overline{L}_i$ bounds the nonlinearities of Eqn. \eqref{Eqn:NormalFormOutput} as $\left|\eta_i - 4E_{t,i}\right|\le \overline{L}_i$ 
	    \item 
	     Assume there exists a non-empty equilibrium set defined as  
	     $\left\{\tilde{x}_i^* \in \mathcal{\tilde{X}}_i  |y_{z,i}\left(\tilde{x}^*_i\right) = P_i^{r,in}\right\}$. 
	    Any point in this set is asymptotically stable. 
    \end{enumerate}
	\end{subequations}
    \label{Theorem:CtrlSMC}
\end{theorem}
\begin{proof}
The proof is detailed in Appendix \ref{Sec:CtrlSMCProof}. 
\end{proof}

The finite settling time in sharp contrast to the asymptotic stability result stated earlier is crucial, especially in electric energy systems, since these systems require components to follow higher-layer supervisory control signals within pre-specified time. 
The FBLC and SMC controllers result in qualitatively similar stability results. However, the characterized bound in SMC $\left|\eta_i - 4E_{t,i}\right| \le \overline{L}_i$ is much more conservative than the differential bound $\left|\tilde{\eta}_i - 4E_{t,i}\right| \le 0 $ in use with FBLC, thereby also leading to increased control effort in the case of former. 
At the same time, perfect tracking is achieved with SMC and is robust to model and parameter uncertainties while FBLC only has exponential convergence and is susceptible to tracking errors that can be upper bounded by the measurement uncertainty. 

We have so far solved the problem of distributed stabilization. 
However, the problem (P1) requires regulation of outputs in conventional space $y_i$ that dictate QoS. 
In order to equip this objective, we propose to revise the map $\psi_i$.
Let us consider QoS dictating variable that appears at the interaction port as an example. The following remark illustrates how such regulation objectives can get accommodated while retaining the claims stated in this section. 
\\
\begin{remark}
(Adopting output regulation objective)\\
Consider the reference for output variables in energy space to be utilized in the map $\psi_i$ in Eqn. \eqref{Eqn:PsiMap} in the component-level control, revised as follows:  
\begin{subequations}
\begin{align}
&{y_{z,i}}^{ref} = {P_i}^{r,in}\left( {\frac{{{y_i}^{ref}}}{{{y_i}}}} \right)\\
&{{\dot y}_{z,i}}^{ref} = {{\dot P}_i}^{r,in}\left( {\frac{{{y_i}^{ref}}}{{{y_i}}}} \right) - {P_i}^{r,in}\left( {\frac{{{y_i}^{ref}}}{{{y_i}^2}}\frac{{d{y_i}}}{{dt}}} \right)
\end{align}
    \label{Eqn:CorrollaryRegulation}
\end{subequations}
All the results derived in this section still hold while the terminal variables get regulated to the desired reference value $y_i^{ref}$. 
\label{Corollary:Regulation}
\end{remark}

In this section, we have explained the different control design methods in energy space that can be utilized to solve problem (P1) from the perspective of the stand-alone component. 
Assuming sufficient internal stability conditions are satisfied at each of the component, we next investigate the feasibility of interconnected system shown in Fig. \ref{Fig:Cutset}. 

\section{Interconnection-level feasibility conditions}
\label{Sec:Feasibility}
In this section,  we propose sufficient feasibility conditions that can be checked for in a feed-forward way to ensure feasible interconnection. 
In closed loop, the energy space model representation observes dynamical model as in Eqn. \eqref{Eqn:IntModelEasy}.
	  \begin{align}
	  &\text{Interaction model in closed loop:}\label{Eqn:IntModelEasy}\\
	   &  \dot{x}_{z,i} = A_{z,i} x_{x,i} + B_t E_{t,i} +B_{z} \dot{z}_i^{r,out}~~ x_{z,i}(0) = x_{z,i0} \notag
	 \end{align}
	 The matrices $A_{z,i}, B_t, B_z$ are the same as defined for Eqn. \eqref{Eqn_IntModelLinear}. 
It should be noted that $\dot{z}_i^{r,out}$ and $E_{t,i}$ here are functions of local energy conversion dynamics, the trajectories of which evolve as per the extended space model considered in Eqn. \eqref{Eqn:InteractiveStdAloneModel}, \eqref{Eqn:ControlImp} and \eqref{Eqn:FBLC}. 
However, the extended state trajectories in the closed loop are dictated by the instantaneous power and generalized reactive power entering the disturbance, control, and interaction ports. 
When two or more components are left to interact, the incoming interaction variable drives the system dynamics. It changes the natural equilibrium of the stand-alone component model to that of the interconnected system models. 
However, this dynamic adjustment at the interfaces is contingent upon the existence of interconnected system equilibrium and the dynamical exchange of power across components. 
The cause of the interactions is the incoming interaction variable, while the effect is the outgoing interacting variable as seen by the component. 
Intuitively the set characterizing the effect should remain within the one characterizing the cause, for convergence. 

Let us first characterize the variation of incoming and outgoing interaction variables over a period of time $t\in \left[(k-1)T, kT\right]$ in the sets $\mathcal{Z}_i^{r,in}[k]$ and $\mathcal{Z}_i^{r,out}[k]$ respectively, where $T >> \delta t$ considered in Theorem \ref{Theorem:CtrlFBLC}. 
At each component, we assume that the set $\mathcal{Z}_i^{r,in}[k]$ is given, for which we characterize $\mathcal{Z}_i^{r,out}[k]$ utilizing the map in Eqn. \eqref{Eqn:IntModelCommon}. 
These sets are communicated to neighbors and are utilized as $\mathcal{Z}_j^{r,out}[k+1]$ for the next time period. 
Numerical algorithms to characterize these sets are a topic of future research. In this paper, we assume these sets can be computed ahead of time and are available to each component. 
We first state a general dissipativity result under certain conditions, that the components can check ahead of time in a distributed manner.  
\\
\begin{lemma}(Feasibility of component interconnection)\\
    If the set $\mathcal{Z}_i^{r,out}[k]$ as characterized by the closed loop interactive model of $\Sigma_i$
    $~ \forall z_i^{r,in}(t) \in \mathcal{Z}_i^{r,in}[k]$
    observes the condition $\mathcal{Z}_i^{r,out}[k] \subseteq \mathcal{Z}_i^{r,in}[k]$, 
	 then, $\Sigma_i$ is dissipative 
	with respect to the supply function $\left(\dot{P}_i^{r,\rm{in}}(t) + \dot{Q}_i^{r,\rm{in}}(t)\right) ~ \forall t \in \left[(k-1)T, kT\right]$.
	\label{Theorem_ComponentDissip}
\end{lemma}
\begin{proof}
The proof is elaborated in Appendix \ref{Sec:FeasibilityProof}. 
\end{proof}

Each of the components satisfying Lemma \ref{Theorem_ComponentDissip} implies that there exists an interconnected system equilibrium $x^*(t)$, which need not be known exactly and is also possibly time-varying because of the time-varying disturbances entering the components as shown in Figure \ref{fig:system_general}. The next stated corollary analyses the stability of this interconnected system created by component interactions through memory-less junctions. \\

\begin{theorem}{(Stability of interconnected system)}\\
Given $\mathcal{Z}_i^{r,in}[k]$ at each component $\Sigma_i$ in the system, assume 
$\Sigma_i$ in closed loop at each time $t$ observes the properties stated in Theorem \ref{Theorem:CtrlFBLC}
$\forall {z}_i^{r,in}(t) \in \mathcal{Z}_i^{r,in}[k]$, and satisfies sufficient feasibility conditions stated in Lemma \ref{Theorem_ComponentDissip}. 
Several such components interacting with each other through memory-less junctions result in an interconnected system that is stable in the sense of Lyapunov $\forall t \in \left[(k-1)T, kT\right]$
under additional assumptions \ref{Assum_Energy} and \ref{Assum_TimeConstant}. 
\label{Corollary_Stability}
\end{theorem}
\begin{proof}
The proof is provided in Appendix \ref{Sec:StabilityProof}, 
\end{proof}

The conditions in Lemma \ref{Theorem_ComponentDissip} are consistent from the perspective of both component-level and system-level. 
They are intended to be checked for in a feed-forward way, which further can be accommodated by a look-ahead centralized controller with dynamical energy space constraints \cite{ilic2020unified} for optimal trajectory planning. However, such an extension is out of the scope of this paper.  

\section{Example of an RLC circuit supplying given power load}
\label{Sec:RLCCkt}
In this section, we take an example of an RLC circuit with a controllable voltage source supplying a given power value as shown in Fig. \ref{fig:system_generalRLC}. This system can be interpreted as two sub-systems interacting similar to the one shown in Fig. \ref{fig:system_general}, where $x_1 = \left[i_1, v_1\right]$ represent the local state variables of source sub-system $\Sigma_1$, $R_1, L_1, C_1$ respectively represent the resistance, inductance and capacitance of controllable sub-system $\Sigma_1$. 
We treat $\Sigma_2$ in this exercise as a black box that is absorbing a given amount of power $P_{l,2}$. 
Since the dynamics of the second component are not modeled, the outgoing interaction variable of $\Sigma_2$ is $\dot{z}_2^{r,out} = \left[{P_{l,2}, \dot{Q}_{l,2}}\right]^T$ where $\dot{Q}_{l,2}$ is the reactive power rate absorbed by the load in order to consume given amount of $P_{l,2}$. 
The incoming rate of interaction variable into $\Sigma_1$ upon interconnection with $\Sigma_2$ is equal to the negative of the rate of outgoing interaction variables of $\Sigma_2$.   
Furthermore, as seen by the stand-alone component $\Sigma_1$, $\dot{z}_1^{r,in}$ can also be interpreted as the disturbance $\dot{z}_1^{m,u} = \left[P_1^{m,u}, \dot{Q}_1^{m,u}\right]^T$ entering $\Sigma_1$.

Here $P_i^{m,u}$ is the disturbance that enters $\Sigma_1$ or the interaction as a result of energy conversion dynamics of $\Sigma_2$. 
The circuit may have additional local disturbances entering near the control port denoted as $P_1^{m,m}$. 
The suffixes $m$ and $u$ indicate the matched and unmatched counterpart of disturbances as defined in \cite{guo2017output}. 
Our proposed controller can stabilize disturbances at both matched and unmatched ports. To restict the scope of this paper, we have considered only the unmatched disturbances in this paper. 
\begin{figure}[!htbp]
\begin{center}
\includegraphics[width=1.0\linewidth]{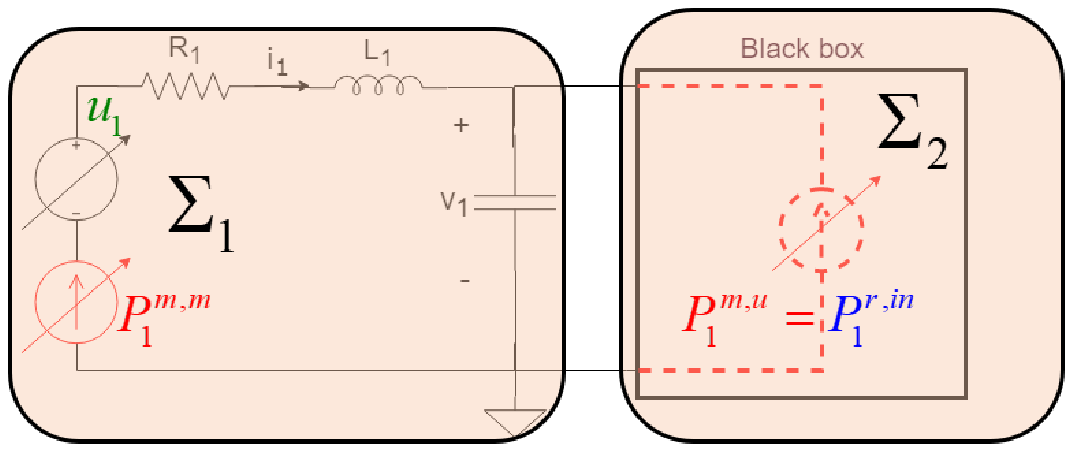}
\caption{RLC circuit with local controllable voltage input $u_1$ ($\Sigma_1$) supplying constant power to a black box with unknown internal dynamics ($\Sigma_2$) when subject to its own local disturbance $P_1^{m,m}$, where the power drawn by $\Sigma_2$, $P_{l,2}$ is interpreted as another disturbance $P_1^{m,u}$ as seen by $\Sigma_1$. }
\label{fig:system_generalRLC}
\end{center}
\end{figure}
The model in standard state space form is given in Eqn. \eqref{Eqn:StateGeneralDyn} for the $\Sigma_1$ is given below: 
\begin{equation}
\begin{array}{l}
\frac{d}{{dt}}\underbrace {\left[ \begin{array}{l}
{i_1}\\
{v_1}
\end{array} \right]}_{{x_1}} = \underbrace {\left[ {\begin{array}{*{20}{c}}
{ - \frac{{{R_1}}}{{{L_1}}}}&{ - \frac{1}{{{L_1}}}}\\
{\frac{1}{{{C_1}}}}&0
\end{array}} \right]\left[ \begin{array}{l}
{i_1}\\
{v_1}
\end{array} \right]}_{{f_{x,1}}({x_1})} + \underbrace {\left[ \begin{array}{l}
\frac{1}{{{L_1}}}\\
0
\end{array} \right]}_{{g_1}^u({x_1})}{u_1} + \\
\qquad \qquad \qquad \qquad \underbrace {\left[ {\begin{array}{*{20}{c}}
{\frac{1}{{{L_1}{i_1}}}}&0\\
0&{\frac{{ - 1}}{{{C_1}{v_1}}}}
\end{array}} \right]}_{{g_1}^m({x_1})}\underbrace {\left[ \begin{array}{l}
{P_1}^{m,m}\\
{P_1}^{m,u}
\end{array} \right]}_{{m_1}}
\end{array}
\label{Eqn:RLCEqn}
\end{equation}
As per the problem posing in (P1) it is desired for the control input $u_1$ to be designed using local measurements to achieve interconnected system stability. Optionally, we required the chosen 
local output variable $y_1 = v_1$ to be regulated to a reference value of $80$ V. 
In this section, we compare our proposed multi-layered distributed control to several other linear and nonlinear methods pursued in the literature.  
Throughout the analysis, we consider 
$R_1 = 10 m\Omega , L_1 = 1.12 mH, C_1 = 6.8 mF$. 

\subsection{Constant power load}
With the interconnection relations, we have $P_1^{m,u} = -P_l = -1.2 kW$.
We first test a conventional proportional control designed as follows: 
\begin{equation}
    u_1 ={u_1^{ref}}  - K_1 i - K_2 \left( v- v^{ref}\right) 
    \label{Eqn:CtrlCG} 
\end{equation}
Here $u_1^{ref} = {v_1^{ref} + R_1 \frac{P_l}{v_1^{ref}}}$ is the consistent input applied at equilibrium. 
The gains $K_1 =0.4512$ and $K_2 = 0.45$ are the Linear Quadratic Regulator (LQR) tuned gains of the linearized system matrices around initial conditions $i_{1,0} = 1A, v_{1,0} = 80 V$. 
Such control results in a stable response as shown in blue in Fig. \ref{Fig:EnergySim3}. 
Notice that implementing this control can be difficult because of having to know the exact reference values and the parameter $R_1$. 

\begin{figure*}[!htbp]
	\centering
	\subfigure[Current through inductor]{\includegraphics[width=0.3\linewidth]{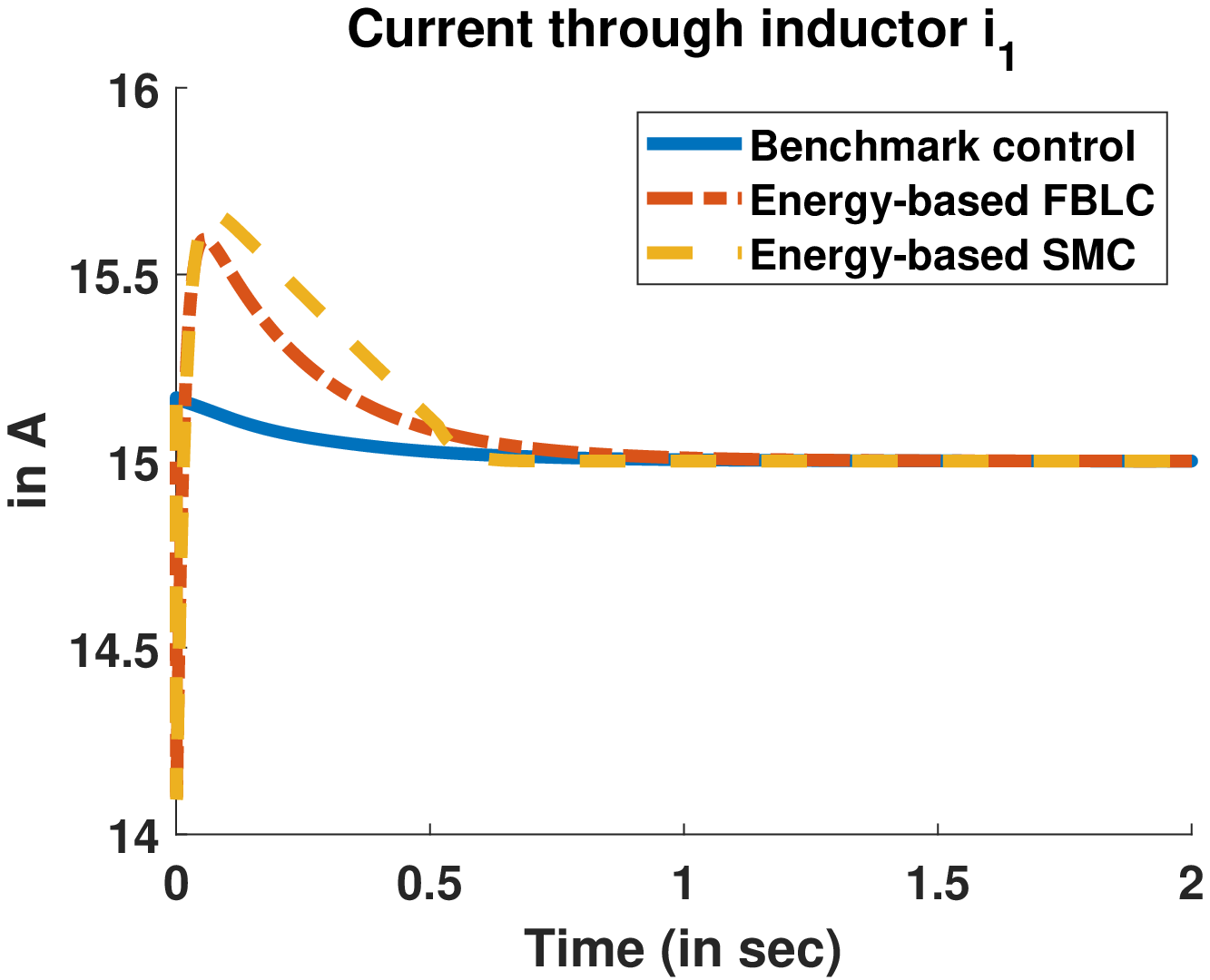}} \label{Fig_Sim3}
	\subfigure[Voltage across capacitor]{\includegraphics[width=0.3\linewidth]{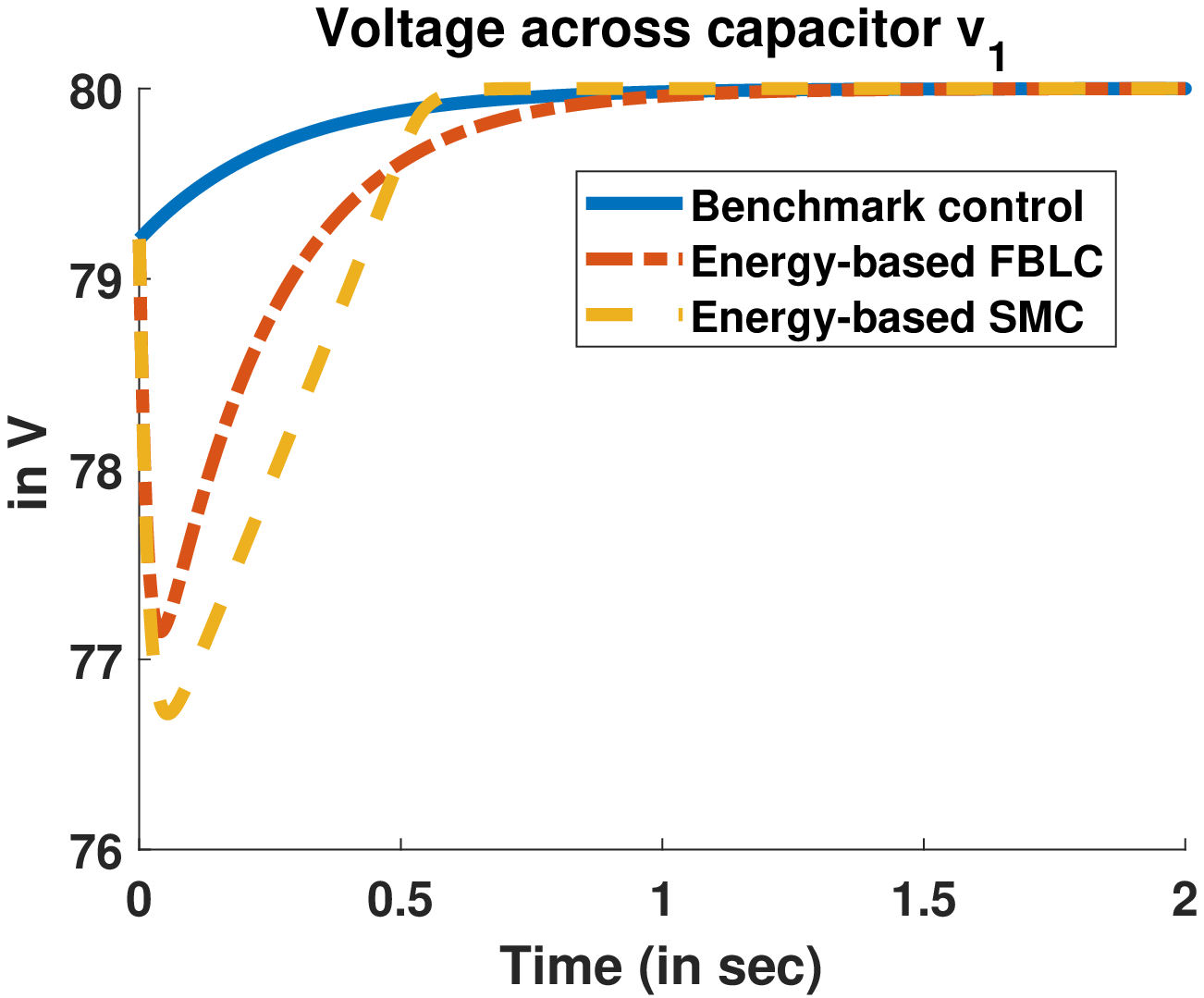}\label{Fig_Voltage_Sim3} }
	\subfigure[Applied control voltage]{\includegraphics[width=0.3\linewidth]{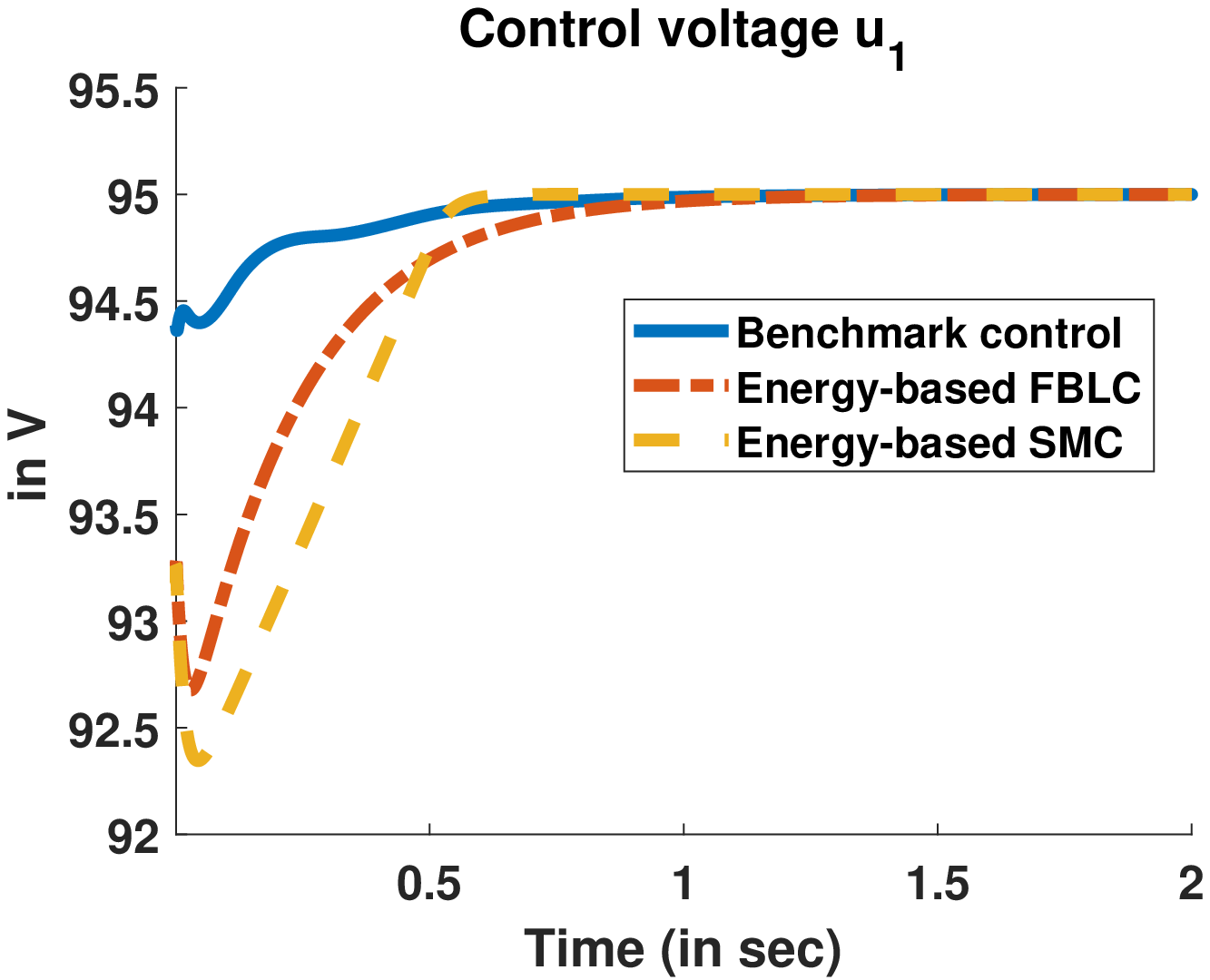}\label{Fig_Control_Sim3} }
	\caption{Comparison of trajectories obtained by applying constant gain control (Eqn. \eqref{Eqn:CtrlCG} in blue) and energy-based control (Equations \eqref{Eqn:ControlDynRLC} in red and yellow), 
	for powering constant power load of $1.2 kW$ with an additional objective to regulate terminal voltage $v_1$ to $80 V$.} 
	\label{Fig:EnergySim3}
\end{figure*}

Now consider the energy based control with $y_{z,1}$ as defined in Eqn. \eqref{Eqn:IntModelOutput_v1}. 
For the system considered in Fig. \ref{fig:system_generalRLC}, 
\begin{subequations}
\begin{equation}
     y_{z,1} = R_1 i_1^2 - u_1 i_1 - P_1^{m,m}
     \label{Eqn:OutputRLC}
\end{equation}
We utilize the higher layer control design as in Eqn. \eqref{Eqn:FBLC} and the mapping in Eqn. \eqref{Eqn:ControlImp},  which for the considered system leads to the following dynamic control
\begin{align}
&{\frac{{d{u_1}}}{{dt}} = \frac{{{u_1}}}{{{i_1}}}\frac{{d{i_1}}}{{dt}} - \frac{{{u_{z,1}}}}{{{i_1}}}}
\label{Eqn:ControlMapRLC}    
\\
&{{\rm{where}}}\notag\\
&{{u_{z,1}} = {\eta _1}\left( {{x_1},{{\dot x}_1}} \right) - {K_1}\left( {{y_{z,1}} - {y_{z,1}}^{ref}} \right) + {{\dot y}_{z,1}}^{ref}}\label{Eqn:EnergyControlRLC} \\
&{\qquad {\eta _1}\left( {{x_1},{{\dot x}_1}} \right) = \underbrace {\begin{array}{*{20}{l}}
{2{L_1}\frac{{d{i_1}}}{{dt}} + }\\
{2{C_1}\frac{{d{v_1}}}{{dt}}}
\end{array}}_{4{E_t}} + 2\underbrace {\left( {\begin{array}{*{20}{l}}
{{v_1}\frac{d}{{dt}}\left( {{i_1} - {i_2}} \right) - }\\
{\left( {{i_1} - {i_2}} \right)\frac{{d{v_1}}}{{dt}}}
\end{array}} \right)}_{{{\dot Q}_c}}}\notag\\
&{\qquad  + \underbrace {2{R_1}{i_1}\frac{{d{i_1}}}{{dt}} - 2{u_1}\frac{{d{i_1}}}{{dt}} + 4{E_{t}}({x_1},{{\dot x}_1})}_{{{\dot P}_1}^{r,out} + {{\dot Q}_1}^{r,out}}}
\label{Eqn:ControlDynRLC_eta}    
\end{align}
   
\begin{figure}[!htbp]
\begin{center}
\includegraphics[width=1.0\linewidth]{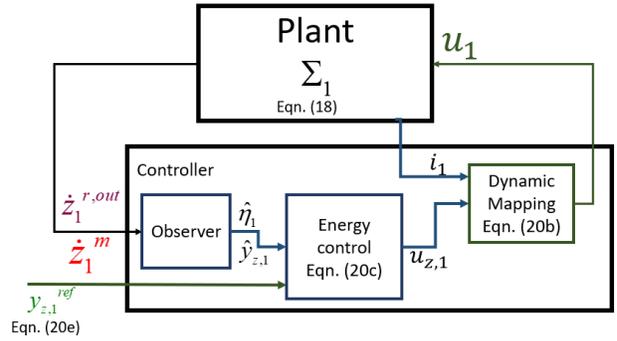}
\caption{Conceptual block diagram of distributed energy control implementation for $\Sigma_1$ of Fig. \ref{fig:system_generalRLC}}.
\label{fig:ControlBlockRLC}
\end{center}
\end{figure}

It is only for simulation runs that we utilize expressions in Eqn. \eqref{Eqn:ControlDynRLC_eta} and \eqref{Eqn:OutputRLC}. 
A conceptual block diagram of control to be utilized in practice is shown in Fig. \ref{fig:ControlBlockRLC}.
It involves usage of an observer that estimates the values of $\eta_1$ and $y_{z,1}$ that would be needed for energy control in Eqn. \eqref{Eqn:EnergyControlRLC}. 
Notably, both observer and control are designed in linear energy state space of Eqn. \eqref{Eqn_IntModelLinear}. In this paper we have only elaborated on the latter and the observer design is a topic of future research.  
Finally, mapping of the designed energy control $u_{z,1}$ to the physical control $u_1$ utilizes 
only real time measurements of conventional state variable $i_1(t)$ at the control port.
The reference values for use in the control design in Eqn. \eqref{Eqn:FBLC} for regulation objective requires revision of the output variable references  in energy space revised as shown in Eqn. \eqref{Eqn:CorrollaryRegulation} which translates to the system under consideration as follows:  
\begin{align}
  &y_{z,1}^{ref} = P_1^{r,in} \frac{v_1^{ref}}{v_1} = -P_{l,2}\frac{v_1^{ref}}{v_1}  \\ 
  &{{\dot y}_{z,1}}^{ref} = {{\dot P}_1}^{r,in}\left( {\frac{{{v_1}^{ref}}}{{{v_1}}}} \right) - {P_1}^{r,in}\left( {\frac{{{v_1}^{ref}}}{{{v_1}^2}}\frac{{d{v_1}}}{{dt}}} \right)\notag\\
   &  =  - {{\dot P}_{l,2}}\left( {\frac{{{v_1}^{ref}}}{{{v_1}}}} \right) + {P_{l,2}}\left( {\frac{{{v_1}^{ref}}}{{{v_1}^2}}\frac{{d{v_1}}}{{dt}}} \right)
   \label{Eqn:OutputRefRLCReg}
\end{align} 
\label{Eqn:ControlDynRLC}    
\end{subequations}
The incoming instantaneous power corresponds to non-zero reactive power absorption characterized in Eqn. \eqref{Eqn:Qinj}, which is the prime cause of non-zero persistent disturbances in conventional state space seen at the interface. 
\begin{equation}
  \dot{Q}_1^{r,in}=-v_2 \frac{d}{dt}\left(\frac{P_{l,2}}{v_2}\right) + \left(\frac{P_{l,2}}{v_2}\right) \dot{v}_2   = - \dot{P}_{l,2} + 2 \frac{P_{l,2} }{v_2} \frac{dv_2}{dt}
  \label{Eqn:Qinj}
\end{equation}

The resulting trajectories with energy based control are shown in Fig. \ref{Fig:EnergySim3} in red and yellow for FBLC and SMC variants respectively.  
It can be seen that the trajectories settle much faster energy-based control and also with tolerable transient overshoots. 
The gain values can further be adjusted based on the desired transient response characteristics. 
Furthermore, notice that the SMC variant clearly indicates finite settling time, while FBLC variant has exponential convergence. 

\subsection{Feed-forward checking of feasibility conditions}
The sufficient feasibility conditions in Eqn. \eqref{Eqn:FeasibilityConds} can be checked in a feed-forward way through pre-specified ranges. 
In this example, since $\Sigma_2$ is assumed a perfect power-consuming load, it can use its estimated variations in power to compute ranges of its outgoing interaction variable, which is sent out as ranges of incoming interaction variable to $\Sigma_1$. 
Given these ranges, $\Sigma_1$ is to compute the ranges of outgoing interaction variables using the definition in \eqref{Eqn:IntModelCommon} given its local control saturation and local predicted disturbance limits. 
If these ranges are within the range of the incoming interaction variable sent out by $\Sigma_1$, only then is the system feasible. 

Consider the load to be served as shown in Fig. \ref{fig:disturbance}. 
\begin{figure}[!htbp]
\begin{center}
\includegraphics[width=0.75\linewidth]{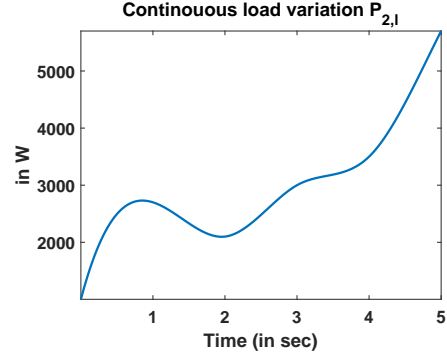}
\caption{Non-zero mean unmatched power disturbances under consideration}
\label{fig:disturbance}
\end{center}
\end{figure}
For this load at every $T_s = 1 s$, we compute the estimated ranges of instantaneous power and rate of change of instantaneous power for use in 
Eqn. \eqref{Eqn:IncomingIntRange} to compute range of incoming interaction variable. 
\begin{subequations}
\begin{align}
&{\rm{Given ~}}{P_{l,2}}(t) \in \left[ {\begin{array}{*{20}{c}}
{{P_{l,2}}^{\min }[k]}&{{P_{l,2}}^{\max }[k]}
\end{array}} \right]\qquad \notag\\
 &\Rightarrow P_1^{r,in}(t) \in \left[ {\begin{array}{*{20}{c}}
{ - {P_{l,2}}^{\max }[k]}&{ - {P_{l,2}}^{\min }[k]}
\end{array}} \right]\\
&{\rm{Given~ }}{{\dot P}_{l,2}}(t) \in \left[ {\begin{array}{*{20}{c}}
{{{\dot P}_{l,2}}^{\min }[k]}&{{{\dot P}_{l,2}}^{\max }[k]}
\end{array}} \right]\qquad \notag\\
 &\Rightarrow {{\dot Q}_1}^{r,in}(t) \in \left[ {\begin{array}{*{20}{c}}
{ - {{\dot P}_{l,2}}^{\max }[k]}&{ - {{\dot P}_{l,2}}^{\min }[k]}
\end{array}} \right]\\
&\quad  + \frac{2}{{{\tau'}}}\left[ {\begin{array}{*{20}{c}}
{{P_{l,2}}^{\min }[k]}&{{P_{l,2}}^{\max }[k]}
\end{array}} \right]\quad \forall t \in \left[ {k{T_s},(k + 1){T_s}} \right] \notag
\end{align}
\label{Eqn:IncomingIntRange}
\end{subequations}
Here $\tau'$ represents the time constant of voltage dynamics to capture transient overshoots. 
$\tau' = min\left\{10 \left(\frac{L_1}{R_1}\right), \frac{1}{K} \right\}$.
$K$ is the effective control gain corresponding to voltage dynamics. 
The way this is computed differs depending on the control design.
For the output feedback control in Eqn. \eqref{Eqn:CtrlCG}, $K = K_2$; and for energy-based FBLC control in Eqn. \eqref{Eqn:FBLC}, $K = K_1$. 

\begin{figure}[!htbp]
	\centering
	\subfigure[First component]{\includegraphics[width=0.45\linewidth]{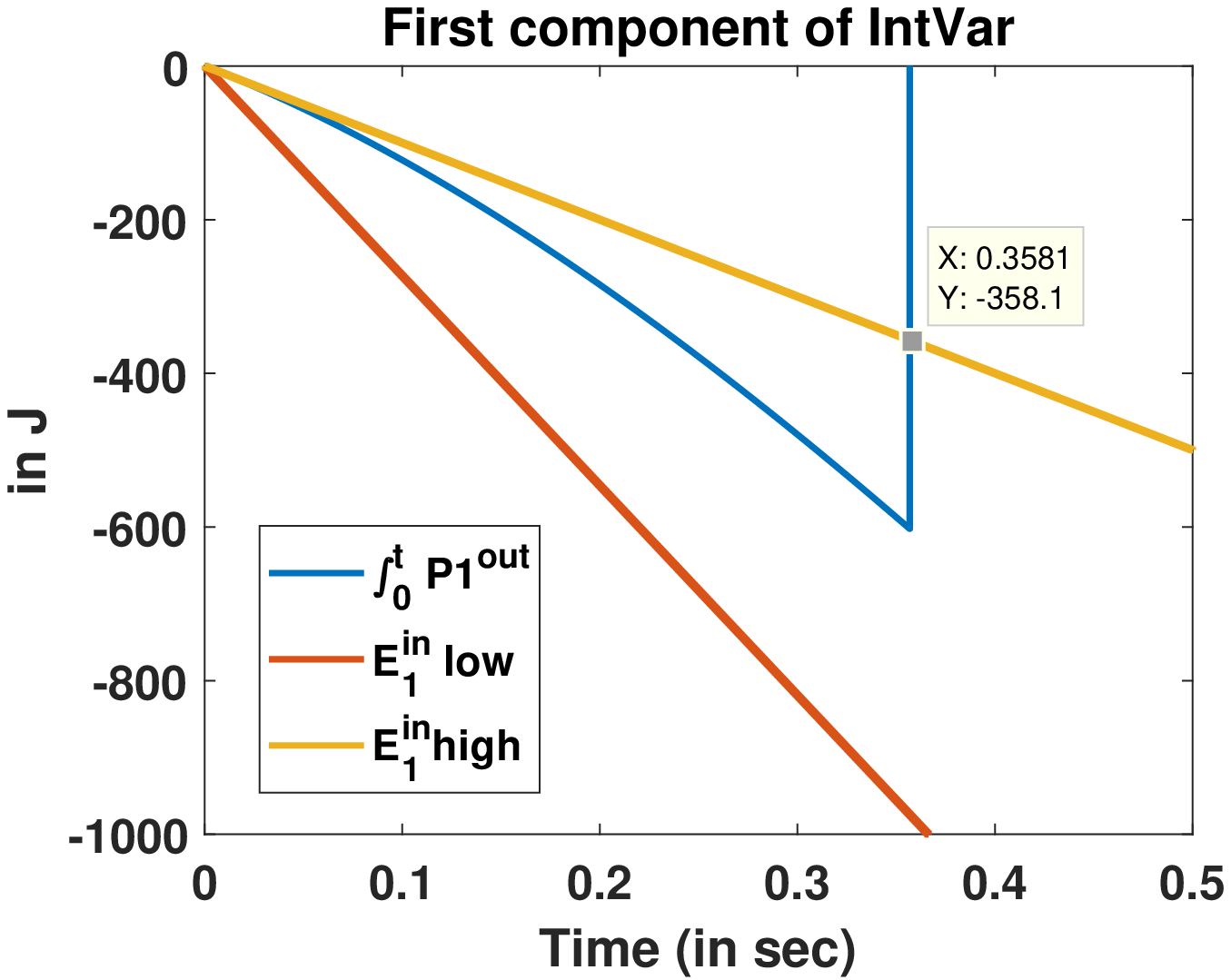}} 
	\subfigure[Second compnent]{\includegraphics[width=0.45\linewidth]{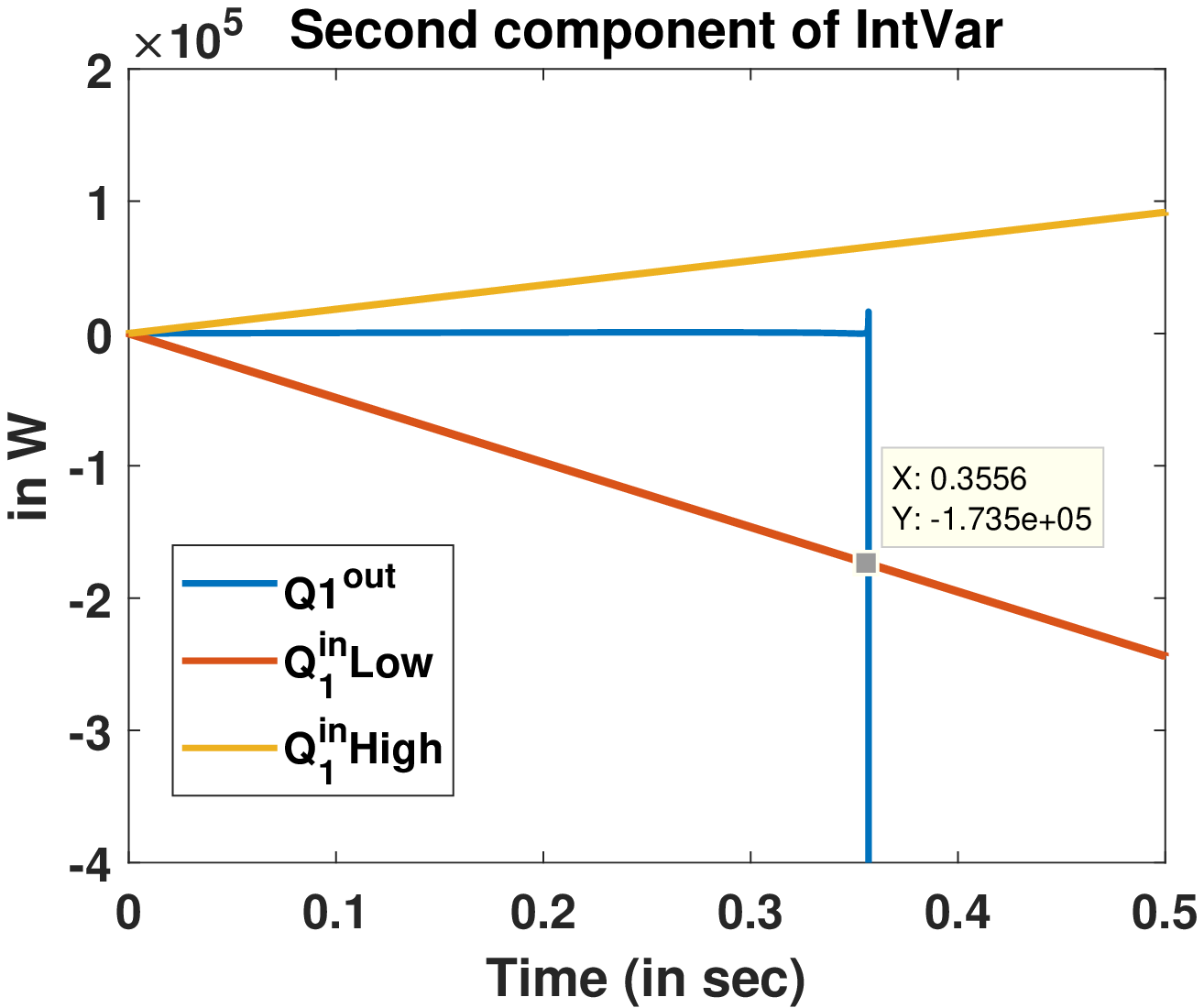}}
	\caption{Interaction variable for validation of feasibility conditions for $\Sigma_1$ of Fig. \ref{fig:system_general} when controlled through output feedback control in Eqn. \eqref{Eqn:CtrlCG}}
	\label{Fig:Feasibility_Unstable}
\end{figure}
The range validation plot for constant gain controllers is shown in Fig. \ref{Fig:Feasibility_Unstable}. 
Notice that the constant gain controller has outgoing reactive power violating the pre-computed incoming reactive power range at approximately 0.35 seconds, indicating infeasibility. 

\begin{figure}[!htbp]
	\centering
	\subfigure[First component]{\includegraphics[width=0.45\linewidth]{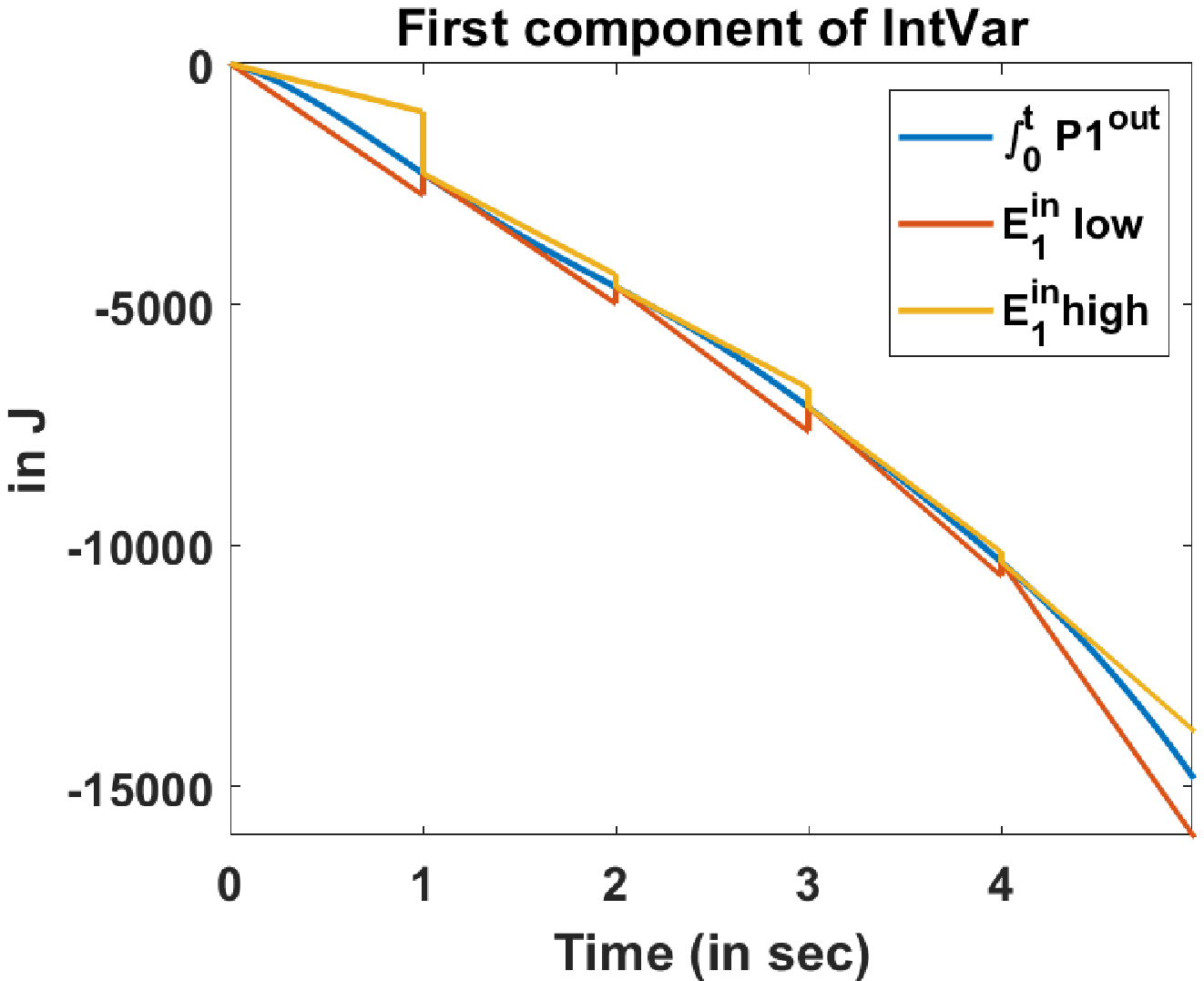}} 
	\subfigure[Second compnent]{\includegraphics[width=0.45\linewidth]{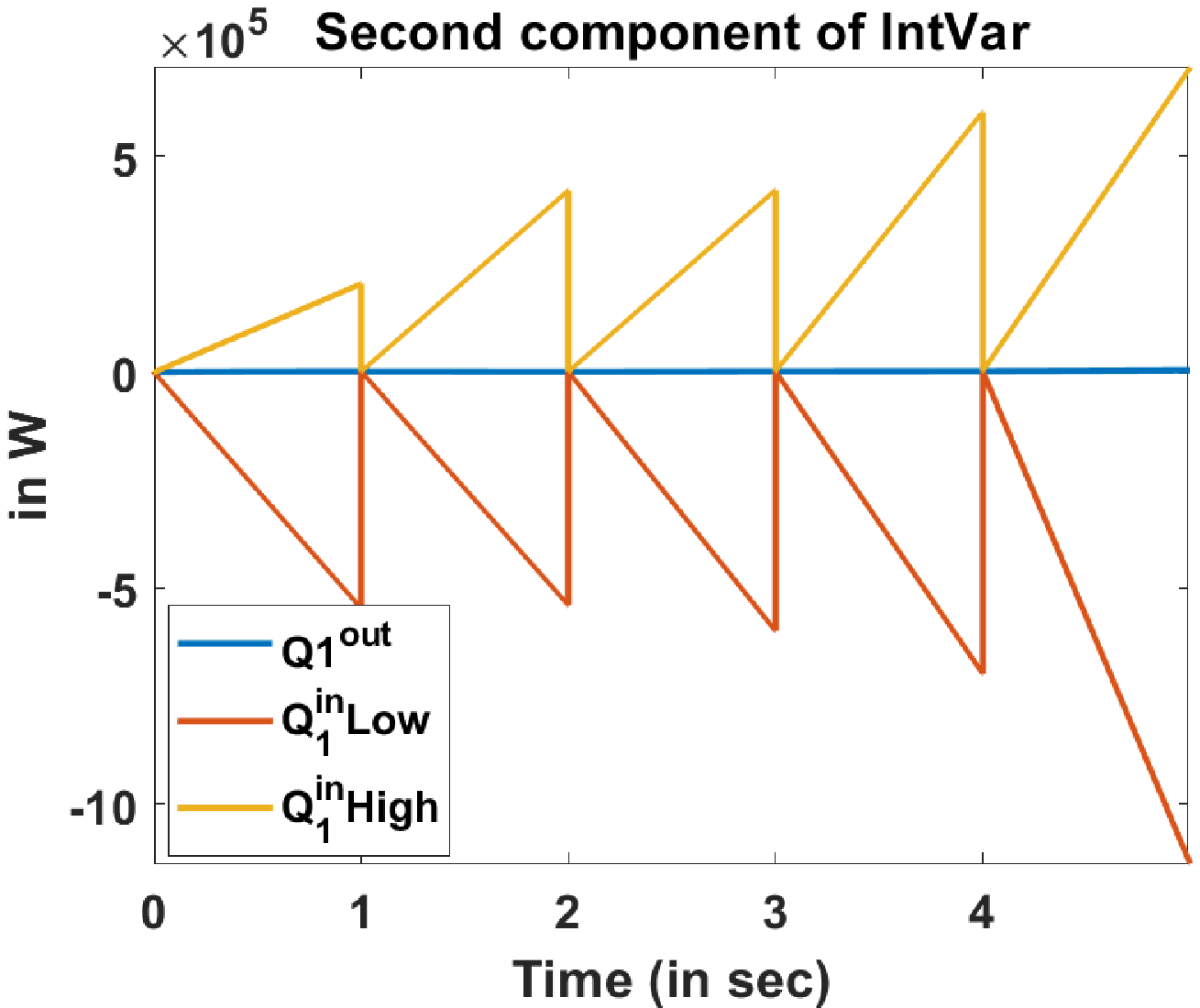}}
	\caption{Interaction variable for validation of feasibility conditions for $\Sigma_1$ of Fig. \ref{fig:system_general} when controlled through design in Eqn. \eqref{Eqn:ControlDynRLC}.}
	\label{Fig:Feasibility_Stable}
\end{figure}
In contrast, energy-based controllers for choice of output variable in Eqn. \eqref{Eqn:OutputRLC} remain stable for the values of load power considered. This is validated by the fact that the outgoing interaction variable always lies within the pre-computed incoming interaction variable range shown in Fig. \ref{Fig:Feasibility_Stable}. 
The corresponding plots of internal variables are shown in Fig. \ref{Fig:EnergySim9} in red. 

\subsection{Performance for continuous time-varying load}
Consider the load to be served as shown in Fig. \ref{fig:disturbance}. This can be perceived as non-zero-mean unmatched disturbances entering the source sub-system. 
\begin{figure*}[!htbp]
	\centering
	\subfigure[Current through inductor]{\includegraphics[width=0.3\linewidth]{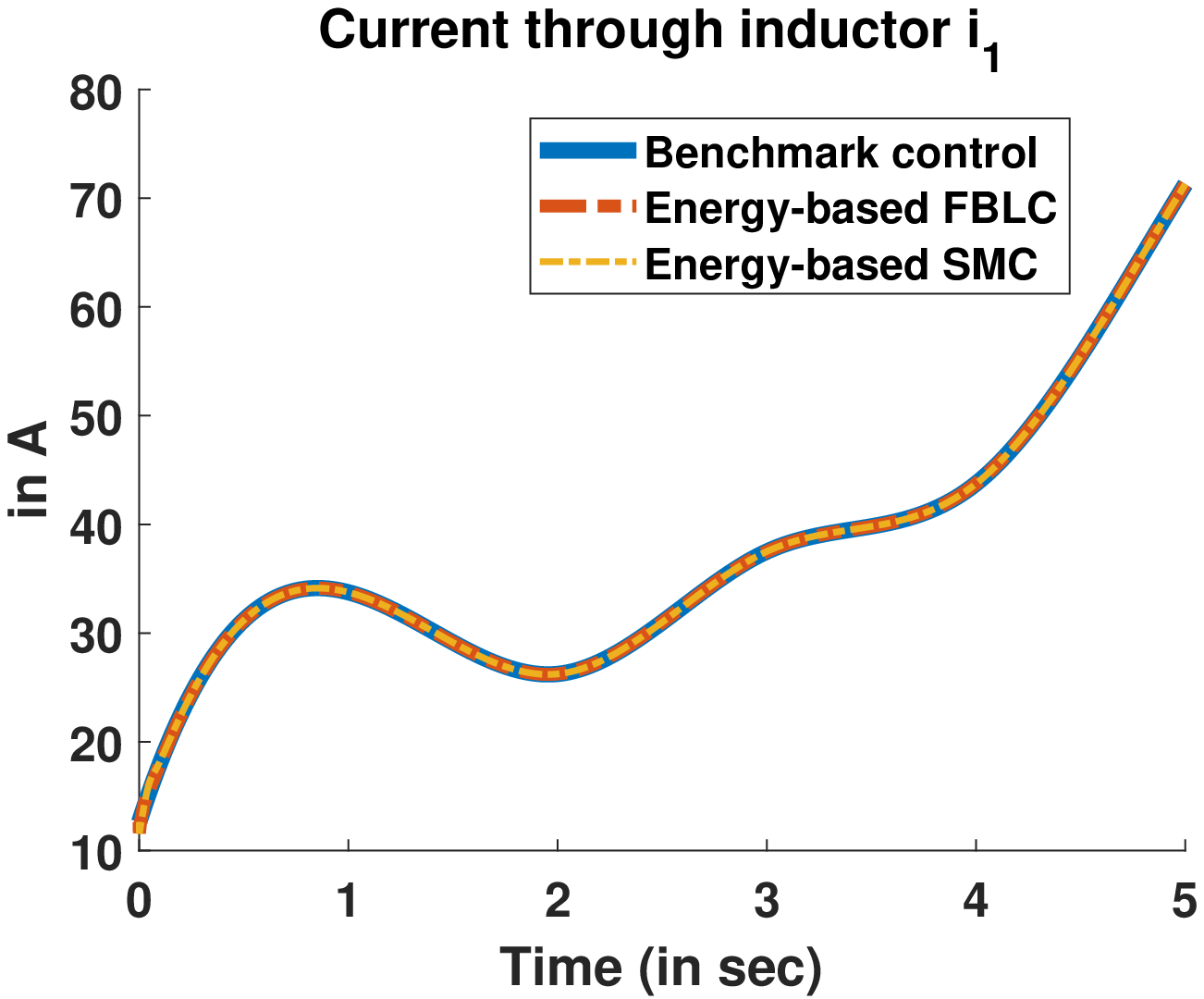}} \label{Fig_Sim9}
	\subfigure[Voltage across capacitor]{\includegraphics[width=0.3\linewidth]{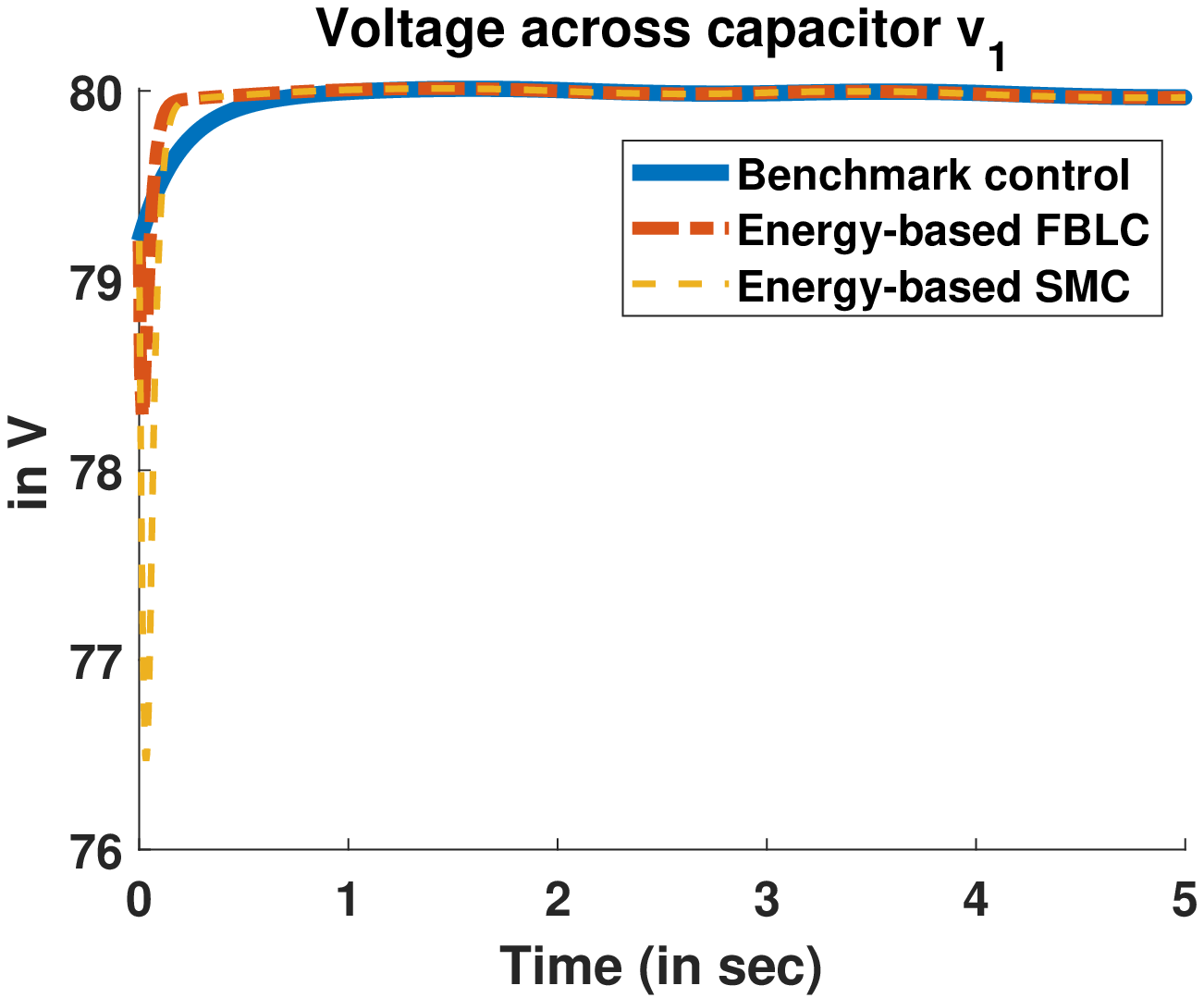}\label{Fig_Voltage_Sim9} }
	\subfigure[Rate of change of control voltage]{\includegraphics[width=0.3\linewidth]{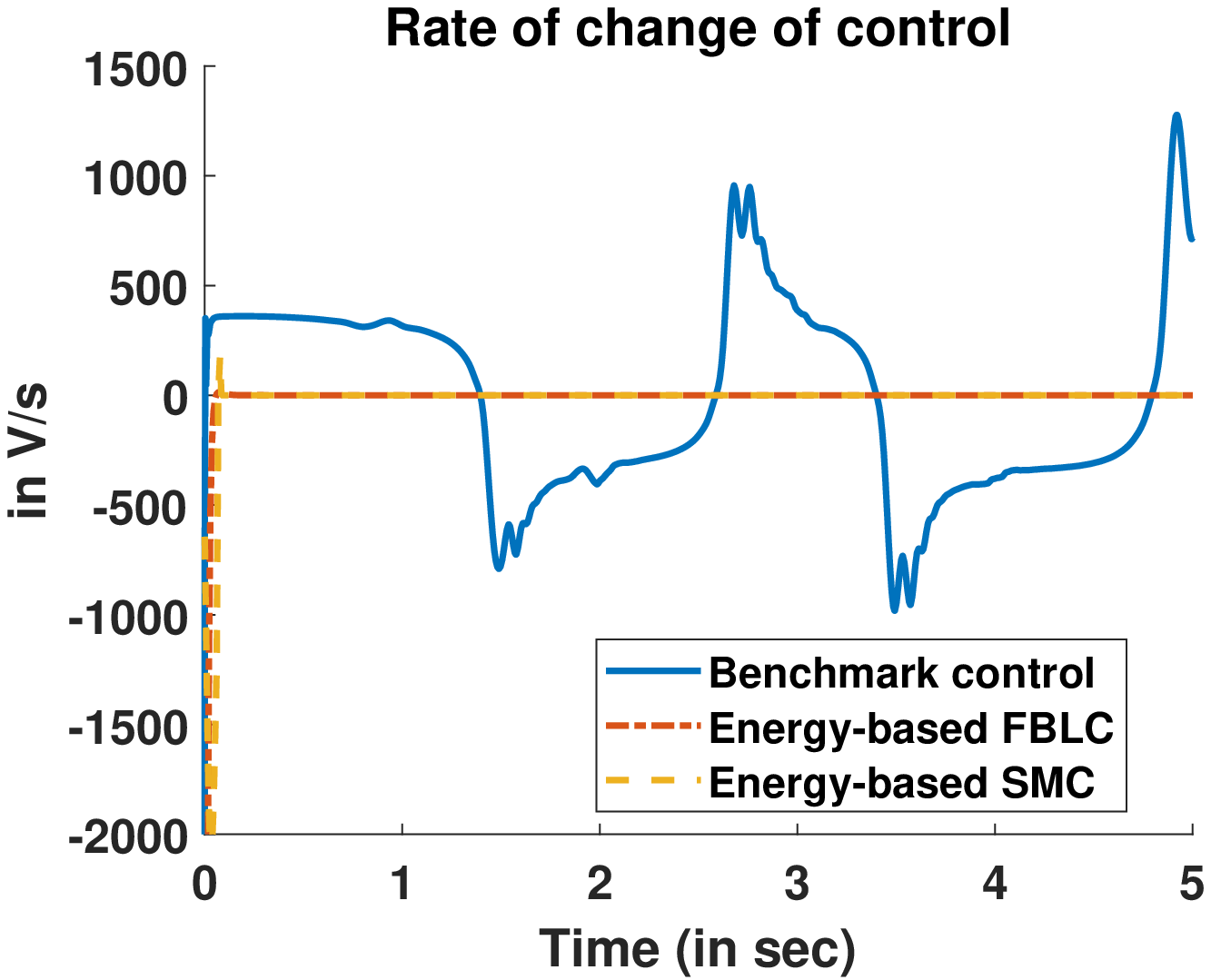}\label{Fig_Control_Sim9} }
	\caption{Comparison of trajectories obtained by applying benchmark nonlinear control (Eqn. \eqref{Eqn:BMCtrl}) and energy-based controllers  (Eqn. \eqref{Eqn:ControlDynRLC})}
	\label{Fig:EnergySim9}
\end{figure*}
The constant gain controllers fail for such time-varying disturbances. 
We thus consider another benchmark nonlinear control proposed in 
\cite{cucuzzella2019voltage} based on Brayton Moser potential function and is shown in Eqn. \eqref{Eqn:BMCtrl}. 
\begin{align}
&\text{Benchmark non-linear control \cite{cucuzzella2019voltage}:}\notag\\
&{u_1} = {R_1}{i_1} +{v_1}^{ref} - {L_1}\left(\frac{\Pi}{{v_1^2}} + K_3\right)\left( {{{\dot v}_1}} \right) \notag\\
& \qquad\qquad- \left( {{K_1}\left( {{v_1} - {v_1}^{ref}} \right) + {K_2}\frac{{d{v_1}}}{{dt}}}\right)
 \label{Eqn:BMCtrl}
\end{align}
Notice that this control utilizes an adaptive gain $\frac{\Pi}{{v_1}^2} + k_3$ for positive $k_3$, where $\Pi$ is the upper bound on the load power values being served. It has been shown in \cite{cucuzzella2019voltage}, the control in Eqn. \eqref{Eqn:BMCtrl} results in stable performance for all values of $v_1^{ref} > 0$ and $P_{l,2} \le \Pi$. 
The corresponding trajectories for the objective of output voltage regulation are shown in Fig. \ref{Fig:EnergySim9} in blue. 

For the same disturbances, the energy-based FBLC with choice of output variable in Eqn. \eqref{Eqn:OutputRLC} results in the response shown in red.
Also shown in yellow is the trajectory obtained by replacing the energy space design $u_{z,i}$ in Eqn. \eqref{Eqn:ControlDynRLC} with  the one in Eqn. \eqref{Eqn:SMC}.
Energy space control explicitly accounts for the interfaces' continuously varying reactive power injections and thus both variants of energy space control
result in better tracking performance in terms of the settling time.
All the controllers result in similar voltage regulation with minimal violations. 
However, the rate of change of control of our proposed method is much smaller than the nonlinear benchmark, leading to much lower wear-and-tear. 
This happens because 
the output variable $y_{z,i}$ to which energy space control responds is an aggregate variable which evolves much slower than conventional state space variables to which the benchmark controller responds to.

\subsection{Robustness to parameter uncertainty}
The nonlinear benchmark control is very sensitive to the parameter $R_1$. 
Shown in Fig. \ref{fig:param_uncertain} are the voltage trajectories obtained when this control is implemented for 10\% error in resistance value measurements. Overlaid on top of it are also the plots obtained with energy-based control. 
With both FBLC and SMC variants, we see that voltage regulation is perfect. This happens because the resistance value is not directly used in the control design. 
Instead, the design utilizes the total damping $\frac{E}{\tau}$ as defined in Eqn. \eqref{Eqn:OutputRLC}, which can be estimated from available measurements of energy and power measurements. 
\begin{figure*}
\subfigure[Voltage]{\includegraphics[width=0.33\linewidth]{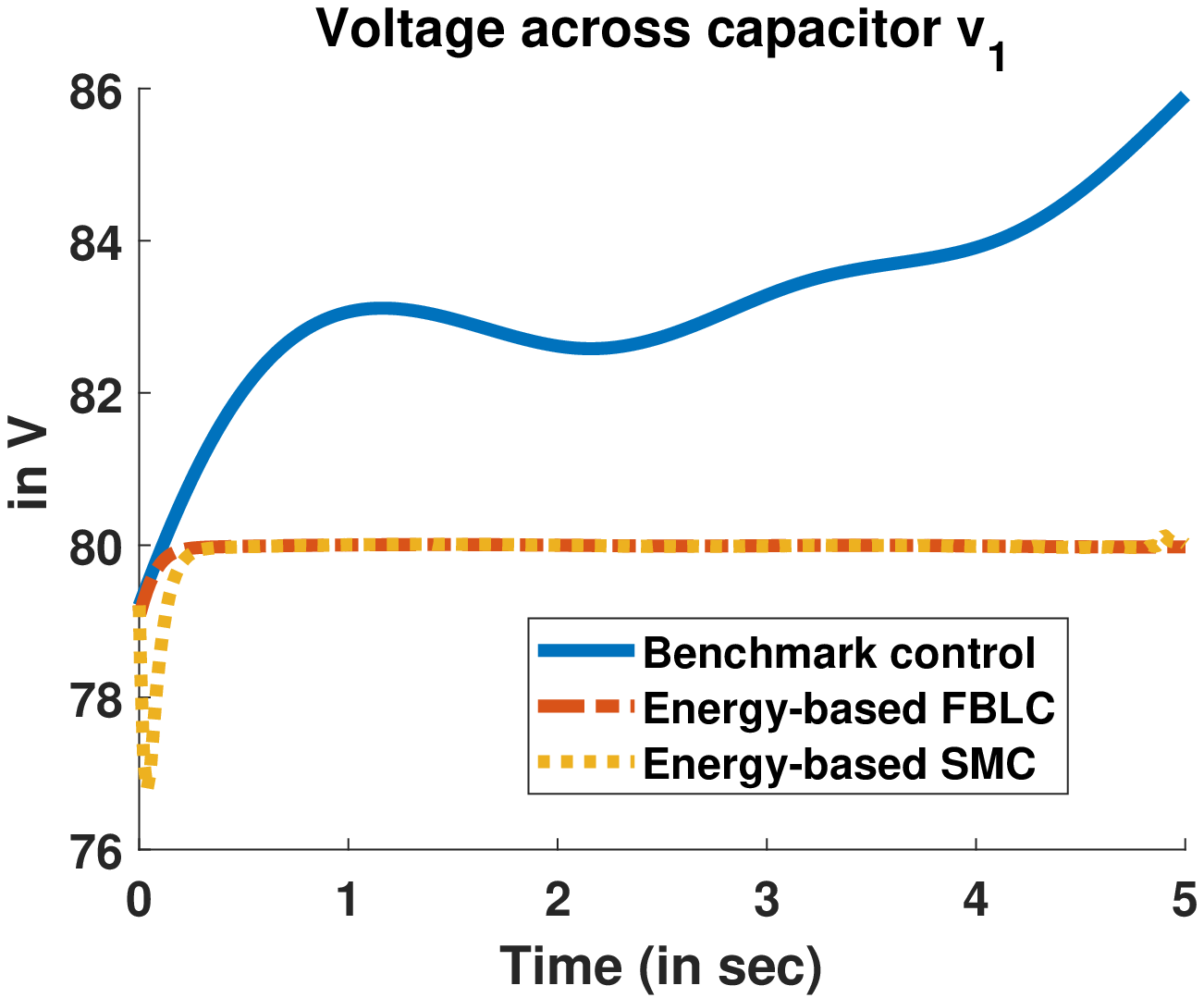}} 
\label{fig:param_uncertain}
\subfigure[Control]{\includegraphics[width=0.33\linewidth]{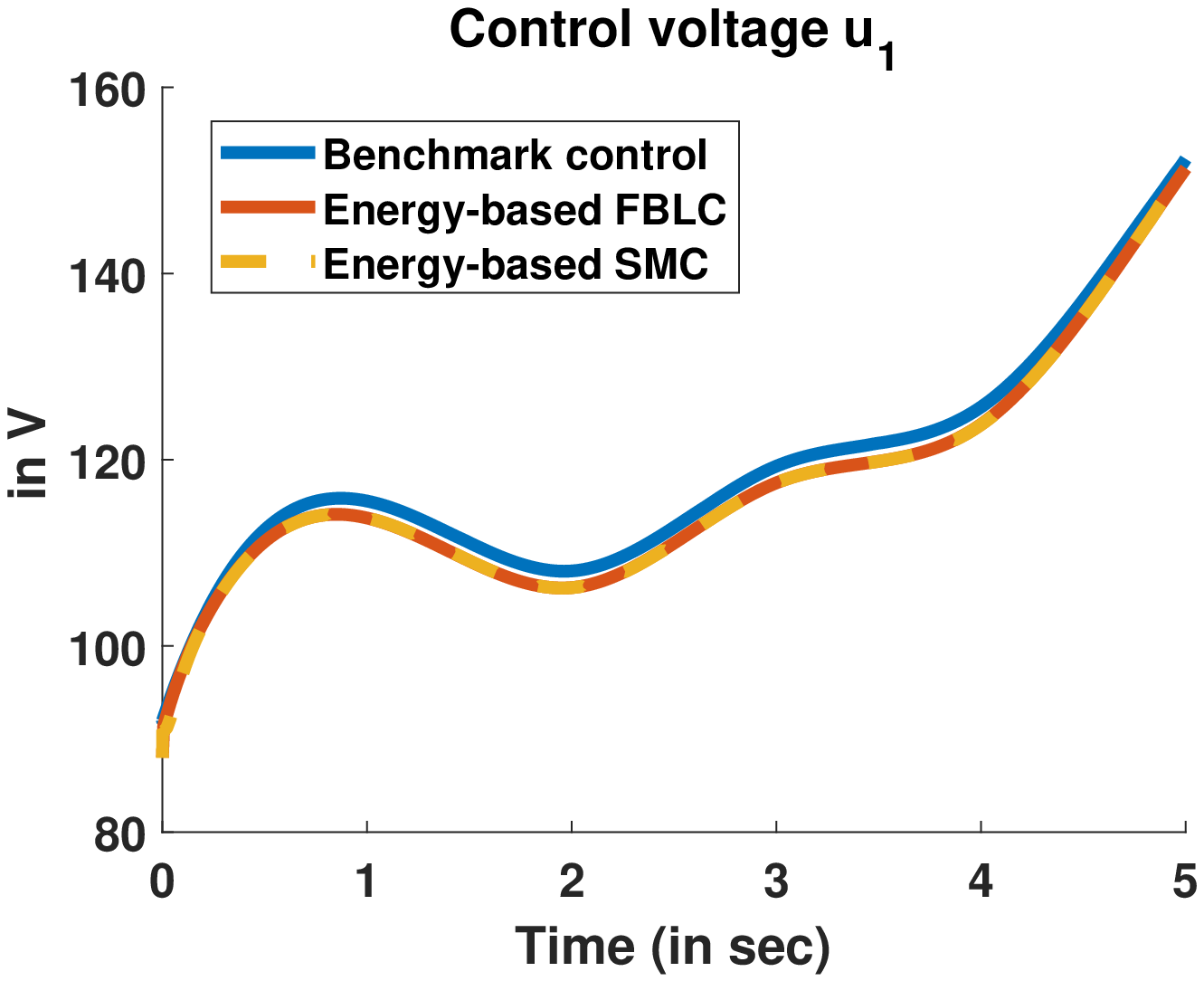}} \label{fig:param_uncertain_u}
	\subfigure[Rate of change of control]{\includegraphics[width=0.33\linewidth]{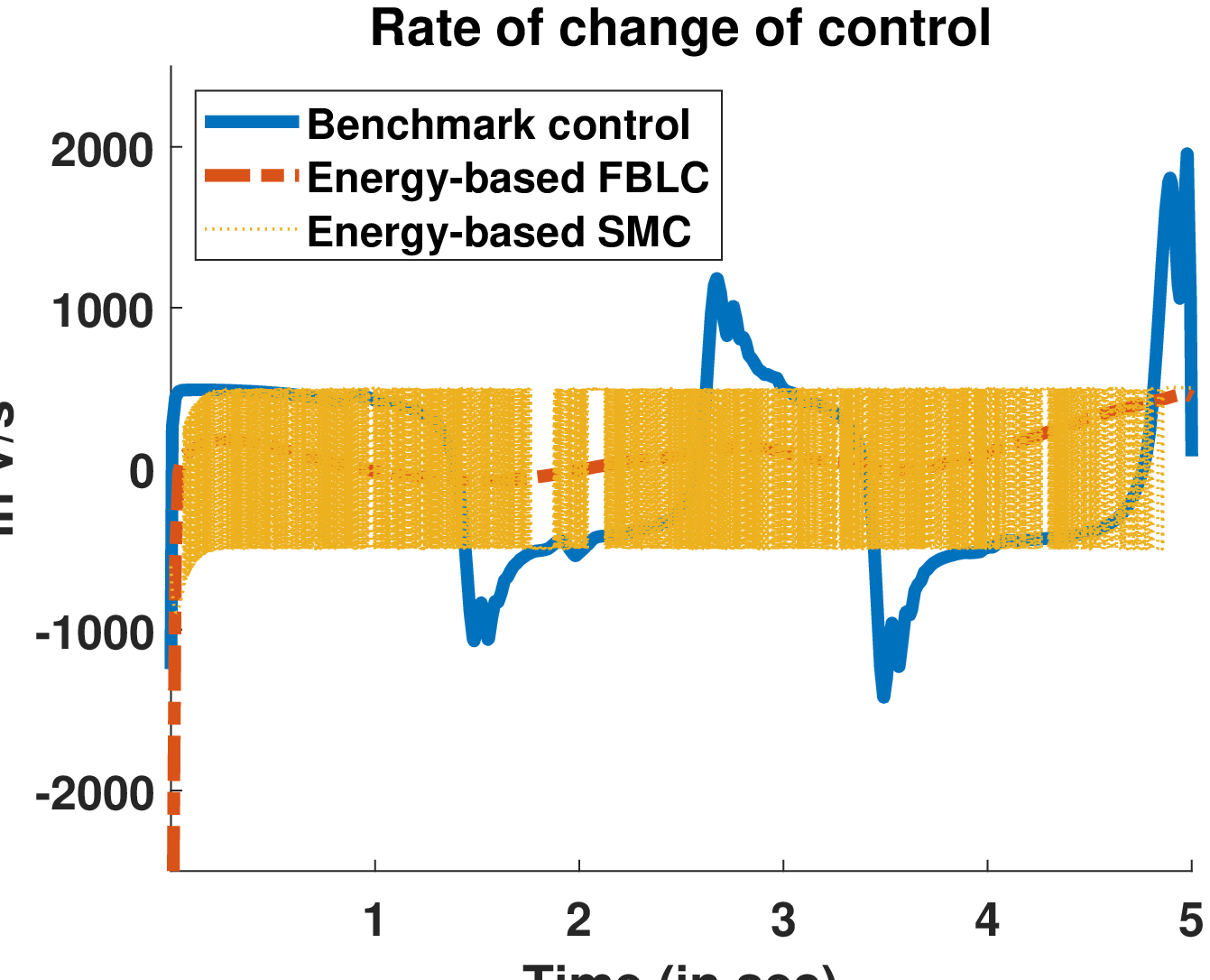}\label{fig:param_uncertain_CtrlRate} }
\caption{Robustness assessment for a time-varying load for 10\% error in resistance value for
benchmark nonlinear control (Eqn. \eqref{Eqn:BMCtrl}) and energy-based controllers  (Eqn. \eqref{Eqn:ControlDynRLC})
}
	\label{Fig:EnergySim10}
\end{figure*}
In presence of such parameter uncertainties, 
the corresponding rate of change of control trajectories of energy-based SMC control is shown in yellow in Fig. \ref{fig:param_uncertain_CtrlRate}. It  indicates clear chattering, while the rate of change of control trajectories are much smoother with energy-based FBLC. Note however that the actual control as in Fig. \ref{fig:param_uncertain_u} is smooth for all the controllers considered. 
It should be noted that in the FBLC variant, we assumed that the disturbance cancellation terms $\eta_i$ in Eqn. \eqref{Eqn:NormalFormOutput} can be measured perfectly, thereby leading to $\tilde{\eta_i} = 0$ in Eqn. \eqref{Eqn:OutputClosed}. Thus the voltage trajectories in Fig. \ref{fig:param_uncertain} appear to be perfect. 
From these experiments, we see that there is a tradeoff between the desired transient response, tolerable wear and tear, measurement and parameter uncertainties in the system. 

\section{Conclusions and future work}
\label{Sec:Conclusion}
In this paper, we consider the problem of distributed stabilization of interconnected systems. We propose an interactive multi-layered energy space modeling,
based on inherent mathematical structures stemming from the first principles of energy conservation laws, offering intuition about energy and power dynamics.
The higher layer models in energy space notably are linear and thus lend themselves to provable control design. 
The control in the lower layer physical model is either FBLC or SMC. 
We show that this multi-layered control design leads to sufficient modular conditions for provable feasibility of the interconnected system and the stability of internal dynamics. 
Notably, in the prior work, models have been utilized for designing competitive disturbance rejection control, which treats the interaction dynamics  as exogenous inputs.
However, these interaction dynamics seen by each component $i$ are state-dependent according to energy conservation principles. 
To capture this interdependence, we propose collaborative control design which instead of cancelling the interactions completely, aligns the references in a cooperative manner. Reference point of component $i$ $y_{z,i}^{ref}$ is obtained by exchanging the information $y_{z,j}$ from neighboring component $j$ resulting in their alignment rather than their cancellation. 
Such interactive models open new avenues for research on distributed cooperative control. 
Finally, we illustrate the effectiveness of the modeling and control design on an RLC circuit
and compare to other linear and nonlinear controllers considered in the literature.

The control design approach is modular and thus is scalable to arbitrarily large systems. However, 
this paper considers an example of single component subject to time-varying power disturbances for exposition of concepts. Further simulation-based evidence of multiple controllable components interacting with each other is needed. 
The interactive energy-based models proposed in this paper pose new challenges in distributed numerical simulations and control implementation. 
Further investigation of the interplay between numerical simulations and the hardware implementation of the control is being pursued.

\begin{ack}                               
This paper is a result of financial support by the NSF project Early-Concept Grants for Exploratory Research (EAGER) ``Fundamentals of Modeling and Control for the Evolving Electric Power System Architectures"   project ECCS-2002570. 
\end{ack}

\bibliographystyle{unsrt}        
\bibliography{autosam}           



\appendices
\section{Energy space variables definitions}
\label{Sec:EnergyDef}
\begin{subequations}
	\begin{definition} (Instantaneous power)\\
		\label{Defn:RealPower}
		The power interaction of the component $i$ with the rest of the system is given by the mapping 
		${P}_i: \mathcal{E}_i \times \mathcal{F}_i \rightarrow \mathcal{P}_i$ and is defined as
		$P_i = e_i^T f_i$ 
		where $e_i \in \mathcal{E}_i$ and $f_i \in \mathcal{F}_i$ respectively represent the effort and flow variables appearing at the ports of interconnection.
	\end{definition}
	\begin{definition}(Generalized reactive power dynamics)\\
		\label{Defn:ReacPower}
		Generalized reactive power $Q_i$ is a quantity whose time derivative is given by a mapping $\dot{Q}_i: \mathcal{TE}_i \times \mathcal{TF}_i \rightarrow \mathcal{TQ}_i$ and is defined as 
		$\dot{Q}_i = e_i^T \frac{df_i}{dt} - f_i^T\frac{de_i}{dt}$ 
		where $\mathcal{TE}_i$ and $\mathcal{TF}_i$ represents the tangent manifold of efforts and flow variables respectively.
	\end{definition}

Let the total stored energy in a generic form for the component defined by the model in Eqn.\eqref{Eqn:StateGeneralDyn} be defined using an inertia matrix $H_i(x_i)$. 
Similarly, if these elements are lossy, they are also associated with a symmetric dissipation matrix $B_i(x_i)$.
\cite{ilic2018multi,jeltsema2009multidomain}. 
	\begin{definition}(Stored energy)\\
		\label{Defn:StoredEnergy}
		Stored energy of component $i$ is given by the energy function ${E}_i: \mathcal{X}_i \rightarrow \mathbb{R}$ defined as 
		$E_i(x_i) = \frac{1}{2}x_i^TH_i(x_i) x_i$ 
		for an inertia matrix $H_i(x_i)~ \forall x_i \in \mathcal{X}_i$.
	\end{definition}
	
	\begin{definition}(Stored energy in tangent space)\\
		\label{Defn:StoredEnergyTangent}
		For the same positive definite matrix $H_i$ as defined in Definition \ref{Defn:StoredEnergy}, the stored energy in tangent space $E_{ti}$ is defined over the tangent bundle $\mathcal{TX}_i := \cup_{x_i \in \mathcal{X}_i} {x_i}\times T_{x_i}\mathcal{X}_i$ where $T_{x_i}\mathcal{X}_i$ is the tangent space of $\mathcal{X}_i$ at $x_i$, through the mapping ${E}_{t,i}: \mathcal{TX}_i \rightarrow \mathbb{R}$ as 
		$E_{t,i}(x_i, \dot{x}_i) = \frac{1}{2}\dot{x}_i^TH_i(x_i) \dot{x}_i$.  
	\end{definition}
	\begin{definition}(Time constant)\\
		\label{Defn:TimeConstant}
		Dissipation of a component $i$ is described through a dissipation function in quadratic form through the mapping ${D}_i: \mathcal{X}_i \rightarrow \mathbb{R}$ defined as $
		D_i(x_i) = x_i^TB_i(x_i)(x_i) x_i $
		for some matrix $B_i(x_i) \forall x_i \in \mathcal{X}_i$. 
		We define the time constant of component $i$ as the ratio of stored energy and the damping of the component, which by utilizing  Eqns. as 
		$\tau_i = \frac{E_i(x_i)}{D_i(x_i)}$ 
		Utilizing the Definitions \ref{Defn:StoredEnergy} and the expression for damping, this quantity can be upper bounded by the largest singular value of $D_i^{-1}(x_i)H_i(x_i)$. 
	\end{definition}
\end{subequations}                                        

\section{Proofs}
\subsection{Theorem \ref{Theorem:CtrlFBLC} of Section \ref{Sec:Control}}
\label{Sec:CtrlFBLCProof}
\begin{proof}
\textit{(1)} 
\begin{subequations}
    Consider now the candidate storage function 
        ${S_i} =  
        y_{z,i}$
    By taking its time derivative and plugging in the expressions from Eqn. \eqref{Eqn:OutputClosed}, we have
\begin{align}
{\frac{{d{S_i}(t)}}{{dt}} = -4{E_{t,i}} + \left( { - {K_i}\left( \begin{array}{l}
{y_{z,i}} - \\
{y_{z,i}}^{ref}
\end{array} \right) + 
{{\tilde \eta }_i} 
+ {{\dot y}_{z,i}}^{ref}} \right)}
\label{Eqn:Temp11}
	\end{align}
	By Assumption \ref{Assum_Energy}, $E_{t,i} \ge 0$. 
	Setting the upper  bound on $\tilde{\eta}_i$ 
	and assigning $y_{z,i}^{ref} = P_i^{r,in}$, the result follows. 

	\textit{(2) }
	Select the closed loop potential function as $V_i = \left|y_{z,i} - y_{z,i}^{ref}\right|$. Taking its time derivative and substituting the derivative expression in Eqn. \eqref{Eqn:OutputClosed}, we obtain
	\begin{equation}
\begin{array}{*{20}{l}}
{\frac{{d{V_i}(t)}}{{dt}} = {\mathop{\rm s}\nolimits} ign\left( \begin{array}{l}
{y_{z,i}} - \\
{y_{z,i}}^{ref}
\end{array} \right)\left( { - {K_i}\left( \begin{array}{l}
{y_{z,i}} - \\
{y_{z,i}}^{ref}
\end{array} \right) + \left( \begin{array}{l}
{{\tilde \eta }_i} - \\
4{E_{t,i}}
\end{array} \right)} \right)}\\
{\quad  \le  - {K_i}\left| {{y_{z,i}} - {y_{z,i}}^{ref}} \right| + \left| {{{\tilde \eta }_i} - 4{E_{t,i}}} \right|}
\end{array}
\label{Eqn:TempFBLCStab}
	\end{equation}
	The first inequality indicates that for arbitrary $\tilde{\eta}_i$ and $4E_{t,i}$, the set defined by $\left\{\tilde{x}_i \in \mathcal{\tilde{X}}_i | y_{z,i}\left(\tilde{x}_i\right) = y_{z,i}^{ref}\right\}$ is the only possibility for $\frac{dV_i}{dt} = 0$. From LaSalle's invariance principle, we thus have asymptotic stability when condition in Eqn. \eqref{Eqn:SuffStabCond} is satisfied. 
	\end{subequations}
\end{proof}

\subsection{Theorem \ref{Theorem:CtrlSMC} of Section \ref{Sec:Control}}
\label{Sec:CtrlSMCProof}
\begin{proof}
	\begin{subequations}
\textit{(1) }
By denoting
	distance of operating point from the sliding surface as $\sigma_i = y_{z,i} - y_{z,i}^{ref}$, consider $S_i = \frac{1}{2}\sigma_i^2$.Taking its time derivative and utilizing Equations \eqref{Eqn:NormalFormOutput} and Eqn. \eqref{Eqn:SMC}, we obtain
	\begin{equation}
	    \begin{array}{l}
{{\dot S}_i} = {\sigma _i}{{\dot \sigma }_i} = {\sigma _i}\left( {\left( { - 4{E_{t,i}} + {\eta _i} + {u_{z,i}}} \right) - {{\dot y}_{z,i}}^{ref}} \right)\\
 = {\sigma _i}\left( { - 4{E_{t,i}} + {\eta _i}} \right) + {\sigma _i}\left( { - {\alpha _i}sign({\sigma _i})} \right)\\
 \le \left| {{\sigma _i}} \right|{{\bar L}_i} - ({{\bar L}_i} + {K_i})\left| {{\sigma _i}} \right|\\
 \le  - {K_i}\left| {{\sigma _i}} \right| \le  - K_i{\left( {2{S_i}} \right)^{\frac{1}{2}}}
\end{array}
	    \label{Eqn:TempSMC}
	\end{equation}
	Taking the time integral on both sides, we can establish 
	\begin{equation}
	\begin{array}{l}
\int\limits_0^t {{S_i}^{ - \frac{1}{2}}{{\dot S}_i}}  \le  - \sqrt 2 \int\limits_0^t {{K_i}dt} \\
\begin{array}{*{20}{l}}
{ \Rightarrow S_i^{\frac{1}{2}}(t) - S_i^{\frac{1}{2}}(0) \le  - \frac{K_i}{\sqrt 2} t}\\
{ \Rightarrow \left| {{\sigma _i}(t)} \right| \le \left| {{\sigma _i}(0)} \right| - \frac{K_i}{\sqrt 2}t}
\end{array}
\end{array}
\label{Eqn:TempSMCDerivation}
	\end{equation}
	Since $|\sigma_i|\ge 0$, we have $\sigma_i(t)$ approach zero in a time upper bounded by $t_r=\frac{\sqrt{2}}{K_i} \left|\sigma_i(0)\right|$
	\end{subequations}
	
	\textit{(2) }
	\begin{subequations}
	From Eqn. \eqref{Eqn:TempSMCDerivation}, 
	 $  \frac{{d{S_i}}}{{dt}} \le -K_i \left| \sigma_i\right|$.
	The right hand sign is negative semi-definite and is zero only when $\sigma_i = 0$ for arbitrary operating conditions and time-varying disturbances. i.e. for when $\tilde{x}_i \in y_{z,i}(\tilde{x}_i) = {P}_i^{r,in}$. 
	By LaSalle's invariance principle, the equilibrium set thus defined in the theorem statement is asymptotically stable. 
	\end{subequations}
\end{proof}

\subsection{Lemma \ref{Theorem_ComponentDissip} of Section \ref{Sec:Feasibility}}
\label{Sec:FeasibilityProof}
\begin{proof}
At instantaneous time, the sufficient feasibility conditions can also be re-stated as an inequality condition to be satisfied element-by-element as follows: 
\begin{equation}
       z_i^{r,out} \preceq z_i^{r,in}
        \label{Eqn:FeasibilityConds}
\end{equation}
	
	\begin{subequations}
		Considering the  derivative of the first element of the inequality and the second element of the inequality and adding them up, we have
		\begin{equation}
		{\dot P_i}^{r,out}+ {\dot Q_i}^{r,out} \le  {\dot P_i}^{r,in}+ {\dot Q_i}^{r,in}
		\label{Eqn_Temp2}
		\end{equation}
		By utilizing the definition of outgoing interaction variable in Eqn. \eqref{Eqn:IntModelEasy}, 
		\begin{align}
		{{\dot P_i}^{r,out}} & = {\dot p_i} + \frac{d}{dt}\left({\frac{1}{{{\tau _i}}}{E_i}} \right) \\
		{\dot Q_i}^{r,out} &= {4{E_{t,i}}}  - {\dot p_i}
		\label{Eqn_Temp3}
		\end{align}
		Now combining Eqns. \eqref{Eqn_Temp2} - \eqref{Eqn_Temp3}, we have
		\begin{equation}
		\frac{d}{dt}{\underbrace {\left(\int\limits_0^t {{{{4{E_{t,i}}(s)}}}ds} + {{\frac{1}{\tau }}_i}{E_i}(t)\right)}_{{S_i(t)}}}  \le  \underbrace{{\dot P_i}^{r,in}  + {\dot Q_i}^{r,in}}_{{s_i (t)}} \label{Eqn:FeasTemp}
		\end{equation}
		The result follows thereby. 
	\end{subequations}	
\end{proof}

\subsection{Theorem \ref{Corollary_Stability} of Section \ref{Sec:Feasibility}}
\label{Sec:StabilityProof}
\begin{proof}
By adding up the dissipativity conditions in Eqn. \eqref{Eqn:FeasTemp} for each of the components in the network $\mathcal{N}$,
\begin{equation}
    \begin{array}{l}
\frac{d}{{dt}}\left( {\sum\limits_{i \in {\mathcal{N}}} {{S_i}({x_i})} } \right) \le \sum\limits_{i \in {\mathcal{N}}} {\left( {\dot P_i^{r,in} + \dot Q_i^{r,in}} \right)} \\
\sum\limits_{i \in {\mathcal{N}}} {\left( {\dot P_i^{r,out} + \dot Q_i^{r,out}} \right)}  \le \sum\limits_{i \in {\mathcal{N}}} { - \left( {\sum\limits_{j \in {{\mathcal{C}}_i}} {\left( {\dot P_j^{r,out} + \dot Q_j^{r,out}} \right)} } \right)} \\
{\bf{1}}_{|N| \times 1}^T\left( \begin{array}{l}
{{{\bf{\dot P}}}^{r,out}} + \\
{{{\bf{\dot Q}}}^{r,out}}
\end{array} \right) \le -{\bf{1}}_{|N| \times 1}^T {\bf{L}}_{\left|\mathcal{N}\right| \times \left|\mathcal{N}\right|} \left( \begin{array}{l}
{{{\bf{\dot P}}}^{r,out}}\\
 + {{{\bf{\dot Q}}}^{r,out}}
\end{array} \right)\\
{\bf{1}}_{|\mathcal{N}| \times 1}^T\left( {\bf{I}_{|\mathcal{N}|\times |\mathcal{N}|} + \bf{L}_{\left|\mathcal{N}\right| \times \left|\mathcal{N}\right|}} \right)\left( {{{{\bf{\dot P}}}^{r,out}} + {{{\bf{\dot Q}}}^{r,out}}} \right) \le 0\\
  \Rightarrow 
  \frac{d}{dt}\left({\bf{1}}_{|\mathcal{N}| \times 1}^T\left( {\bf{I}_{|\mathcal{N}|\times |\mathcal{N}|} + \bf{L}_{\left|\mathcal{N}\right| \times \left|\mathcal{N}\right|}} \right)\bf{S}\right) \le 0
\end{array}
    \label{Eqn:StabTemp}
\end{equation}
Here, $\bf{P}, \bf{Q}$ are the vector form notation of the real and reactive power of each of the components in the network.
$|.|$ operator here represents the cardinality of the set. 
$\bf{I}$ and $\bf{1}$ respectively represent the identity matrix and the column vector comprising element $1$,  with its order in the subscript.
$\bf{L}$ represents a symmetric matrix where the element $L_{ij} = 1$ is the components $i$ and $j$ are connected. 
Finally, $\mathbf{S}$ is the vectorized representation of potential functions considered for each of the components. 
Since each element of $\mathbf{S}$ is positive definite due to the assumptions \ref{Assum_Energy} and \ref{Assum_TimeConstant}, we have that the entire bracketed term in the last equation to be positive definite. This bracketed term can be considered as the candidate Lyapunov function to thereby prove stabilitiy in the sense of Lyapunov. 
\end{proof}

\end{document}